\documentclass[a4paper,11pt]{article}
\usepackage{microtype}

\usepackage{multirow}
\usepackage{amsmath,amssymb,amsthm}
\usepackage{algorithm, algorithmic}
\newtheorem{claim}{Claim}
\newtheorem{Reduction Rule}{\bf Reduction Rule}
\newtheorem{observation}{\bf Observation}

\usepackage{enumerate}

\usepackage{xspace}

\newcommand{\NP}{\text{\normalfont  NP}}
\newcommand{\CoNP}{\text{\normalfont co-NP}}
\newcommand{\FPT}{\text{\normalfont FPT}}

\newcommand{\W}[1][xxxx]{\text{\normalfont W[#1]}}

\newcommand{\F}{\ensuremath{\mathcal F}\xspace}

\newcommand{\Oh}{\mathcal{O}}

\newcommand{\SC}{\textsc{Set Cover}}
\newcommand{\PLC}{\textsc{Point Line Cover}}
\newcommand{\MCC}{\textsc{Multicolored Clique}}

\newcommand{\oRBSC}{\textsc{Red Blue Set Cover}}
\newcommand{\sorbsc}{\textsc{RBSC}}
\newcommand{\RBSC}{\textsc{Generalized Red Blue Set Cover}}
\newcommand{\srbsc}{\textsc{Gen-RBSC}}
\newcommand{\LRBSC}{\textsc{Generalized Red Blue Set Cover with lines}}
\newcommand{\slrbsc}{\textsc{Gen-RBSC-lines}}
\newcommand{\oLRBSC}{\textsc{Red Blue Set Cover with lines}}
\newcommand{\solrbsc}{\textsc{RBSC-lines}}

\newcommand{\wslrbsc}{\textsc{Weighted Gen-RBSC-lines}}
\newcommand{\GI}{\textsc{Subgraph Isomorphism}}

\newcommand{\YES}{\textsc{YES}}
\newcommand{\NO}{\textsc{NO}}

\usepackage{amsfonts}
\usepackage{amstext}
\usepackage{graphicx,tikz}
\RequirePackage{fancyhdr}
\usepackage{xcolor}
\usepackage{boxedminipage}
\newcommand{\defparproblem}[4]{
  \vspace{1mm}
\noindent\fbox{
  \begin{minipage}{0.96\textwidth}
  \begin{tabular*}{\textwidth}{@{\extracolsep{\fill}}lr} #1  & {\bf{Parameter:}} #3
\\ \end{tabular*}
  {\bf{Input:}} #2  \\
  {\bf{Question:}} #4
  \end{minipage}
  }
  \vspace{1mm}
}

\newcommand{\defproblem}[3]{
  \vspace{1mm}
\noindent\fbox{
  \begin{minipage}{0.96\textwidth}
  \begin{tabular*}{\textwidth}{@{\extracolsep{\fill}}lr} #1 \\ \end{tabular*}
  {\bf{Input:}} #2  \\
  {\bf{Question:}} #3
  \end{minipage}
  }
  \vspace{1mm}
}

\usepackage{todonotes}

\theoremstyle{plain}
\newtheorem{theorem}{\bf Theorem}[section]
\newtheorem{lem}{\bf Lemma}[section]
\newtheorem{prop}{\bf Proposition}
\newtheorem{definition}{Definition}
\newtheorem{corollary}{Corollary}
\newtheorem{thmk}{\bf Theorem}

\usepackage{vmargin}
\setmarginsrb{1.1in}{1.1in}{1.1in}{1.1in}{0pt}{0pt}{0pt}{7mm}
 \usepackage[ pdftex, plainpages = false, pdfpagelabels, 
                 bookmarks=false,
                 bookmarksopen = true,
                 bookmarksnumbered = true,
                 breaklinks = true,
                 linktocpage,
                 pagebackref,
                 colorlinks = true,  
                 linkcolor = blue,
                 urlcolor  = blue,
                 citecolor = red,
                 anchorcolor = green,
                 hyperindex = true,
                 hyperfigures
                 ]{hyperref}

\title{Multivariate Complexity Analysis of Geometric {\sc Red Blue Set Cover} }
 \author{Pradeesha Ashok \thanks{Institute of Mathematical Sciences, Chennai, India
 				Email :\texttt{pradeesha@imsc.res.in}}
 \and Sudeshna Kolay\thanks{Institute of Mathematical Sciences, Chennai, India
 				Email :\texttt{skolay@imsc.res.in}}
\and Saket Saurabh\thanks{Institute of Mathematical Sciences, Chennai, India
				Email :\texttt{saket@imsc.res.in}}}

\date{}
\begin{document}
\maketitle
\begin{abstract}
We investigate the parameterized complexity of \RBSC\ (\srbsc), a generalization of the classic \SC\ problem and the more recently studied \oRBSC\ problem. Given a universe $U$ containing $b$ blue elements and $r$ red elements, positive integers $k_\ell$  and $k_r$, and a family $\F$ of $\ell$ sets over $U$, the \srbsc\ problem is to decide whether there is a subfamily $\F'\subseteq \F$ of size at most $k_\ell$ that covers all blue elements, but at most $k_r$ of the red elements. This 
generalizes  \SC\  and thus  in full generality it is intractable in the parameterized setting. In this paper, we study a geometric version of this problem, called \slrbsc, where the elements are points in the plane and sets are defined by lines. 
We study this problem for an array of parameters, namely, $k_\ell, k_r, r, b$, and  $\ell$, and all possible combinations of them. For all these cases, we either prove that the problem is W-hard or show that the problem is fixed parameter tractable (\FPT). In particular, on the algorithmic side, our study shows that a combination of  $k_\ell$ and $k_r$ gives rise to a nontrivial algorithm for \slrbsc. On the hardness side, we show that  the problem is para-\NP-hard when  parameterized by $k_r$, and \W[1]-hard when parameterized by $k_\ell$. Finally, for the combination of parameters for which \slrbsc\ admits \FPT\ algorithms, we ask for the existence of polynomial kernels. We are able to provide a complete kernelization dichotomy  by either showing that the problem admits a polynomial kernel or that it does not contain a polynomial kernel unless $\CoNP \subseteq \NP/\mbox{poly}$.

\end{abstract}
\section{Introduction}

The input to a covering problem consists of a universe $U$ of size $n$, a family $\mathcal F$ of $m$ subsets of $U$ and a positive integer $k$, and the objective is to check whether there exists a subfamily ${\mathcal{F'}}\subseteq {\mathcal F}$  of size at most $k$ satisfying some desired properties.  If  ${\cal F}'$ is required to contain all the elements of $U$, then it corresponds to the classical {\sc Set Cover} problem.  The \SC\ problem is part of Karp's $21$ \NP-complete problems \cite{Karp10}. This, together with its numerous variants, is one of the most well-studied problems in the area of algorithms and complexity. It is one of the central problems in all the paradigms that have been established to cope with 
\NP-hardness, including approximation algorithms, randomized algorithms and parameterized complexity. 

  

%
%
%
%
%
\subsection{Problems  Studied, Context and Framework} The goal of this paper is to study a generalization of a variant of \SC\, namely, the \oRBSC\ problem. 

 \defproblem{{\oRBSC}  (\sorbsc)}{A universe $U=(R,B)$ where $R$ is a set of $r$ red elements and $B$
is a set of $b$ blue elements, a family $\F$ of $\ell$ subsets of $U$, and a positive integer $k_r$.}{Is there a
subfamily ${\F'}$ of sets that covers all blue elements but at most $k_r$ red elements?}

\noindent 
\oRBSC\ was introduced in 2000 by Carr et al.~\cite{CDKM00}. This problem is closely related to several combinatorial optimization problems such as the {\sc Group Steiner}, {\sc Minimum Label Path}, {\sc Minimum Monotone Satisfying Assignment}  and {\sc Symmetric Label Cover} problems. This has also found applications in areas like fraud/anomaly detection, information retrieval and the classification problem. \oRBSC\  is \NP-complete, following from an easy reduction from \SC\ itself. 

In this paper, we study the parameterized complexity, under various parameters, of a  common generalization of both 
\SC\ and \oRBSC, in a geometric setting.

\defproblem{{\RBSC} (\srbsc)}{A universe $U=(R,B)$ where $R$ is a set of $r$ red elements and $B$
is a set of $b$ blue elements, a family $\F$ of $\ell$ subsets of $U$, and positive integers $k_\ell,k_r$.}{Is there a
subfamily ${\F'}\subseteq \F$ of size at most $k_\ell$ that covers all blue elements but at most $k_r$ red elements?}

\noindent 
It is easy to see that when $k_\ell = \vert \F\vert$ then the problem instance is a \oRBSC\ instance, while it is a \SC\ instance when $k_\ell = k, R = \emptyset, k_r = 0$. Next we take a short detour and give a few essential definitions regarding parameterized complexity. 

\medskip 

\noindent 
{\bf Parameterized complexity.} The goal of parameterized complexity is to find ways of solving NP-hard problems more efficiently than brute force: here the aim is to
restrict the combinatorial explosion to a parameter that is hopefully
much smaller than the input size. Formally, a {\em parameterization}
of a problem is assigning a positive integer parameter $k$ to each input instance and we
say that a parameterized problem is {\em fixed-parameter tractable
  (\FPT)} if there is an algorithm that solves the problem in time
$f(k)\cdot \vert I \vert ^{\Oh(1)}$, where $\vert I \vert$ is the size of the input and $f$ is an
arbitrary computable function depending only on the parameter $k$. Such an algorithm is called an \FPT\ algorithm and such a running time is called \FPT\ running time. There is also an accompanying theory of parameterized intractability using which one can identify parameterized problems that are unlikely to admit \FPT\ algorithms. These are essentially proved by showing that the problem is W-hard. 
A parameterized problem is said to admit a $h(k)$-{\it kernel} if there is a polynomial time algorithm (the degree of the polynomial is independent of $k$), called a {\em kernelization} algorithm, that reduces the input instance to an instance with size upper bounded by $h(k)$, while preserving the answer. If the function $h(k)$ is polynomial in $k$, then we say that the problem admits a polynomial kernel.  While positive kernelization results have appeared regularly over the last two decades, the first results establishing infeasibility of polynomial kernels for specific problems have appeared only recently. In particular, 
Bodlaender et al.~\cite{BodlaenderDFH09}, and Fortnow and Santhanam~\cite{FortnowS11} have developed a framework for showing that a problem does not admit a polynomial
kernel unless $\CoNP \subseteq \NP/\mbox{poly}$, which is deemed unlikely. 
For more background, the reader is referred to the following monograph \cite{FlumGrohebook}.

In the parameterized setting, \SC, parameterized by $k$, is \W[2]-hard \cite{DF99} and it is not expected to have an \FPT\ algorithm. The \NP-hardness reduction from \SC\ to \oRBSC\ implies that \oRBSC\ is  \W[2]-hard parameterized by the size $k_\ell$ of a solution subfamily. However, the  hardness result was not the end of the story for the \SC\ problem in parameterized complexity. In literature, 
various special cases of \SC\ have been studied. A few examples are instances with sets of bounded size~\cite{FellowsKNRRSTW08}, sets with bounded intersection~\cite{LP05,RS08}, and instances where the bipartite incidence graph corresponding to the set family has bounded treewidth or excludes some graph $H$ as a minor~\cite{DemaineFHT05,FominGSS09}. Apart from these results,  there has also been extended study on different  parameterizations of \SC. 
A special case of \SC\  which is central to the topic of this paper is the one where the {\em sets in the family correspond to 
some geometric object}.  In the simplest geometric variant of \SC, called \PLC, the elements of $U$ are points in ${\mathbb R}^2$ and each set contains a maximal number  of collinear points. This version of the problem is \FPT\ and in fact has a polynomial kernel \cite{LP05}. Moreover, the size of these kernels have been proved to be tight, under standard assumptions, in~\cite{KPR14}.
When we take the sets to be the space bounded by unit squares, \SC\  is \W[1]-hard~\cite{Mar05}. On the other hand  when surfaces of hyperspheres are sets then the problem is \FPT~\cite{LP05}. There are several other geometric variants of \SC\ that have been studied in parameterized complexity, under the parameter $k$, the size of the solution subfamily. These geometric results motivate a systematic study of the parameterized complexity of geometric \srbsc\ problems. 

There is an array of natural parameters in hand for the \srbsc\ problem. Hence, the problem promises an interesting dichotomy in parameterized complexity, under the various parameters. In this paper, we concentrate on the \LRBSC\ problem, parameterized under combinations of natural parameters.

 \defproblem{{\LRBSC } (\slrbsc)}{A universe $U=(R,B)$ where $R$ is a set of $r$ red points and $B$
is a set of $b$ blue points, a family $\F$ of $\ell$ sets of $U$ such that each set contains a
maximal set of collinear points of $U$, and positive integers $k_\ell,k_r$.}{Is there a subfamily ${\F'}\subseteq \F$ of
size at most $k_\ell$ that covers all blue points but at most $k_r$ red points?} \vspace{10 pt}
 
\noindent 
It is safe to assume that $r \geq k_r$, and $\ell \geq k_\ell$. Since it is enough to find a minimal solution family $\F'$, we can also assume that $b \geq k_\ell$.


We finish this section with some related results.
As mentioned earlier, the \oRBSC\ problem in classical complexity is \NP-complete. Interestingly, if the incidence matrix, built over the sets and elements, has the consecutive ones property then the problem is in $P$ \cite{DGNW07}. The problem has been studied in approximation algorithms as well \cite{CDKM00,Peleg07}. Specially, the geometric variant, where every set is the space bounded by a unit square, has a polynomial time approximation scheme (PTAS)~\cite{CH13}.
\subsection{Our Contributions}
In this paper, we first show a complete dichotomy of the parameterized complexity of \slrbsc. For a list of parameters, namely, $k_\ell, k_r, r, b$, and  $\ell$, and all possible combinations of them, we show hardness or an \FPT\ algorithm. Further, for parameterizations where an \FPT\ algorithm exists, we either show that the problem admits a polynomial kernel or that it does not contain a polynomial kernel unless $\CoNP \subseteq \NP/\mbox{poly}$.

To describe our results we first state a few definitions. For a set $S\subseteq U$, we denote by $2^S$ the family of all the subsets of $S$, and by $U^S$ the family of all the subsets of $U$ that contain $S$ (that is, all supersets of $S$ in $U$). For a collection $\mathfrak{F}$ of sets over a universe $U$, by ${\sf DownClosure}(\mathfrak{F})$ and ${\sf UpClosure}(\mathfrak{F})$  we mean the families $ \bigcup_{S\in \mathfrak{F}}  2^S  \mbox{ and } \bigcup_{S\in \mathfrak{F}} U^S  $ respectively. Our first contribution is the following parameterized and kernelization dichotomy result for \slrbsc. 

\begin{theorem}\label{thm1}
\label{thm:main:rbscalg+kernel}
Let $\Gamma=\{\ell,r,b,k_\ell,k_r\}$. Then \slrbsc\ is \FPT\  parameterized by $\Gamma'\subseteq \Gamma$ if and only if 
$\Gamma'  \notin {\sf DownClosure}(\{ \{k_\ell,b\},\{r\}\})$. Furthermore,   
  \slrbsc\ admits a polynomial kernel parameterized by $\Gamma'\subseteq \Gamma$ if and only if 
 $\Gamma'  \in {\sf UpClosure}(\{ \{\ell\}, \{k_\ell,r\},\{b,r\}\})$.
\end{theorem}
Essentially, the theorem says that if \slrbsc\ is \FPT\  parameterized by $\Gamma'\subseteq \Gamma$ then there exists an algorithm for \slrbsc\ running in time $f(\Gamma')\cdot (\vert U\vert+ \vert {\cal F}\vert)^{\Oh(1)}$. That is, the running time of the algorithm can depend in an arbitrary manner on the parameters present in $\Gamma'$. Equivalently, we have an algorithm running in time 
 $f(\tau)\cdot (\vert U \vert + \vert {\cal F}\vert )^{\Oh(1)}$, where $\tau=\sum_{q\in \Gamma'}q$.  Similarly, if the problem admits a polynomial kernel parameterized by $\Gamma'$ then in polynomial time we get an equivalent instance of the problem of size 
 $\tau^{\Oh(1)}$. On the other hand when we say that the problem does not admit polynomial kernel parameterized by 
 $\Gamma'$ then it means that there is no kernelization algorithm outputting a kernel of size $\tau^{\Oh(1)}$ unless  $\CoNP \subseteq \NP/\mbox{poly}$.  A schematic diagram explaining the results proved in  Theorem~\ref{thm1} can be seen in Figure~\ref{fig:hierarchyebds}. Results for a $\Gamma'\subseteq \Gamma$ which is not 
 depicted in Figure~\ref{fig:hierarchyebds} can be derived by checking whether $\Gamma'$ is in ${\sf DownClosure}(\{ \{k_\ell,b\},\{r\}\})$.

 
 \begin{figure}[t]
\begin{center}
\includegraphics[scale=0.4]{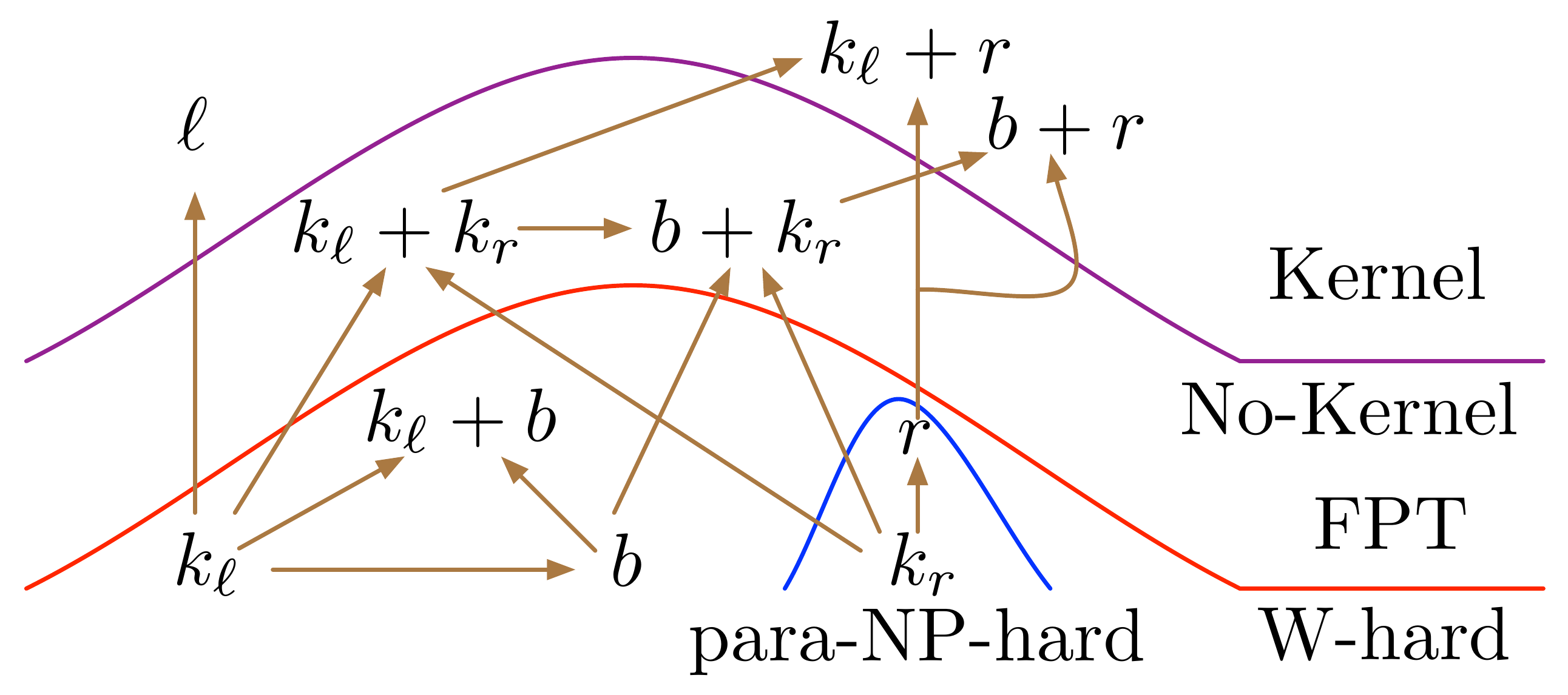}
\scalebox{.55}{
}
\caption{Illustration of our results described in Theorem~\ref{thm1} and hierarchy of parameters. }
\label{fig:hierarchyebds}
\end{center}
\end{figure}
 
 
Next we consider the \solrbsc\ problem. 
Here we do not have any constraint on how many sets we pick in the solution family but we are allowed to cover at most $k_r$ red points. This brings two main changes in Figure~\ref{fig:hierarchyebds}. For \slrbsc\ we show that the problem is \NP-hard even when there is a constant number of red points. However, \solrbsc\ becomes \FPT\ parameterized by $r$. In contrast, \solrbsc\ is \W[1]-hard parameterized by $k_r$.  This leads to the following dichotomy theorem for \solrbsc. 

\begin{theorem}\label{thm2}
\label{thm:main:oldrbscalg+kernel}
Let $\Gamma=\{\ell,r,b,k_r\}$. Then \solrbsc\ is \FPT\  parameterized by $\Gamma'\subseteq \Gamma$ if and only if 
$\Gamma'  \notin \{ \{b\},\{k_r\}\})$. Furthermore,   
 \solrbsc\ admits polynomial kernel parameterized by $\Gamma'\subseteq \Gamma$ if and only if 
 $\Gamma'  \in {\sf UpClosure}(\{ \{\ell\},\{b,r\}\})$.
\end{theorem}
A schematic diagram explaining the results proved in Theorem~\ref{thm2} is given in Figure~\ref{fig:hierarchyeolrbsc}. 
\begin{figure}[t]
\begin{center}
\scalebox{.55}{
\includegraphics[height=8cm]{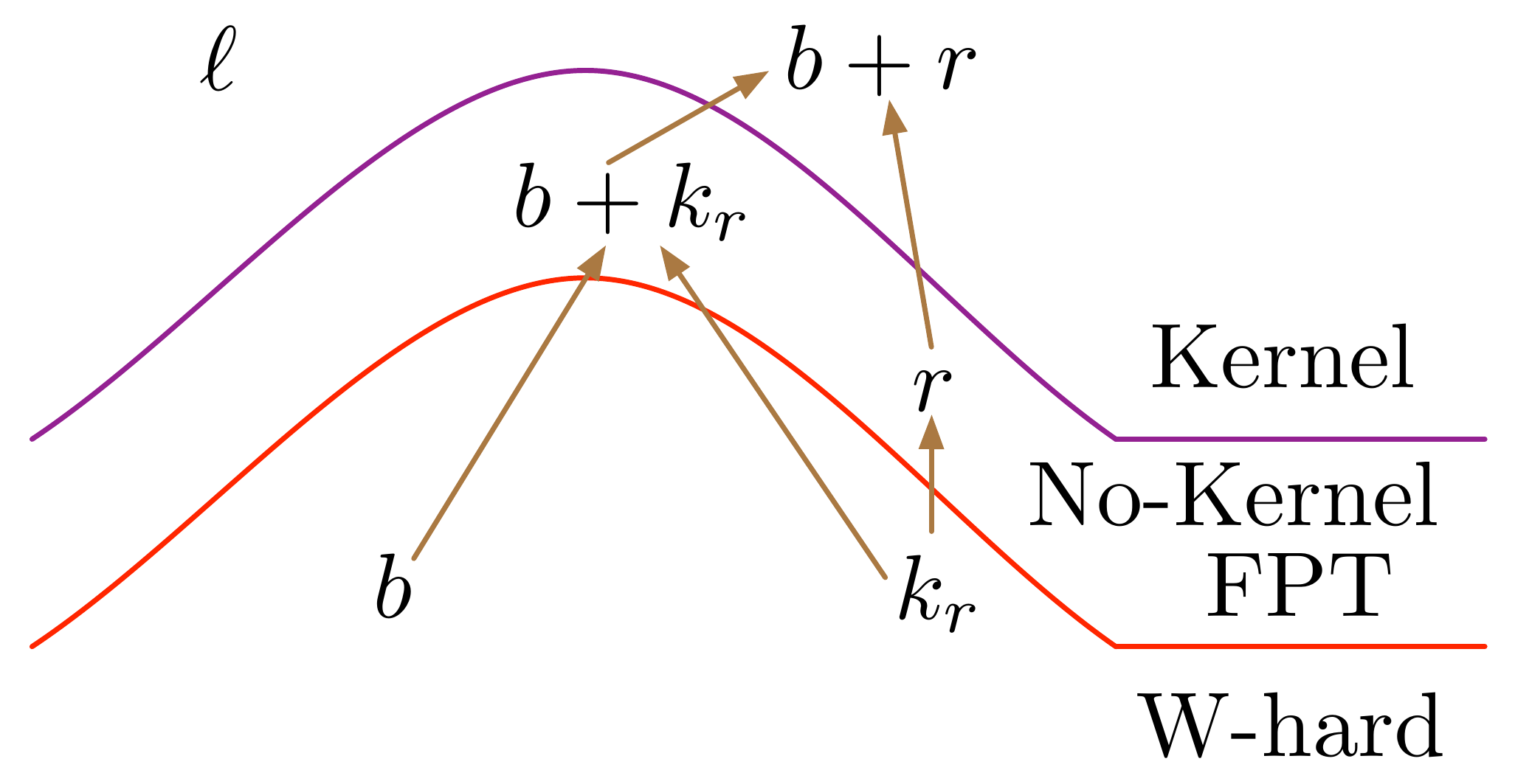}
}

\caption{Illustration of our results for \oLRBSC\ under various parameters. }
\label{fig:hierarchyeolrbsc}
\end{center}
\end{figure}

A quick look at Figure~\ref{fig:hierarchyebds} will show that the \slrbsc\ problem is \FPT\  parameterized by $k_\ell+k_r$ or 
$b+k_r$. A natural question to ask is whether \srbsc\ itself (the problem where sets in the input family are arbitrary and do not correspond to lines) is \FPT\ when parameterized by $b+k_r$. Regarding this, we show the following results:

\begin{enumerate}
\item \srbsc\ is \W[1]-hard parameterized by $k_\ell+k_r$ (or $b+k_r$) when every set has size at most three and contains at least two red points. 
\item \srbsc\ is \W[2]-hard parameterized by $k_\ell+r$ when every set contains at most one red point. 
\end{enumerate}
The first result essentially shows that  \srbsc\  is  \W[1]-hard even when the sets in the family has size {\em bounded by three}. This is in sharp contrast to \SC, which is known to be \FPT\ parameterized by $k_\ell$ and $d$. Here, $d$ is 
the size of the maximum cardinality set in $\cal F$. In fact, \SC\ admits a kernel of size $k_\ell^{\Oh(d)}$. This leads to the following question: 
\begin{quote}
Does the hardness of  \srbsc\ in item one arise from the presence of two red points in the instance?  Would the complexity change if we assume that each set contains  at most one red point? 
\end{quote}
In fact, even if we assume that each set contains at most one red point, we must take $d$, the size of the maximum cardinality set in $\cal F$, as a parameter. Else, this would correspond to the hardness result presented in item two. As a final algorithmic result we show that \srbsc\ admits an algorithm with running time $2^{\Oh(dk_\ell)}\cdot (\vert U \vert +\vert {\cal F} \vert)^{\Oh(1)}$, when every set has at most one red point.  Observe that in this setting $k_r$ can always be assumed to be less than $k_\ell$. Thus, this is also a \FPT\ algorithm parameterized by $k_\ell+k_r$, when sets in the input family are bounded. However, we show that \srbsc\  (in fact \slrbsc) does not admit a polynomial kernel parameterized by $k_\ell+k_r$ even when each set  in the input family corresponds to a line and has size two and contains at most one red point.


\subsection{Our methods and an overview of main algorithmic results}
Let $\Gamma=\{\ell,r,b,k_\ell,k_r\}$. 
Most of our W-hardness results for a \srbsc\ variant parameterized by  $\Gamma'\subseteq \Gamma$ are obtained by giving a polynomial time reduction, from \SC\ or \MCC\, that makes every $q\in \Gamma'$ at most $k^{\Oh(1)}$ (in fact most of the time $\Oh(k)$). 
This allows us to transfer the known  hardness results about \SC\ and \MCC\ to our problem. Since in most cases the parameters are linear in the input parameter, in fact we can rule out an algorithm of form $(\vert U \vert +\vert {\cal F} \vert )^{o(\tau)}$, where 
$\tau=\sum_{q\in \Gamma'}q$, under Exponential Time Hypothesis (ETH)~\cite{ImpagliazzoPZ01}. 
Similarly, hardness results for kernels are derived from giving an appropriate polynomial time reduction from parameterized variants of the \SC\ problem 
that only allows each parameter $q\in \Gamma'$ to grow polynomially in the input parameter.

Our main algorithmic highlights are parameterized algorithms for
\begin{enumerate}
\item[(a)]    \slrbsc\  running in time 
$2^{\Oh(k_\ell \log k_\ell +  k_r \log k_r)}\cdot (\vert U \vert +\vert {\cal F}\vert )^{\Oh(1)}$ (showing  \slrbsc\  is \FPT\ 
parameterized by $k_\ell+k_r$); and 
\item[(b)]   \srbsc\ 
with running time $2^{\Oh(dk_\ell)}\cdot (\vert U \vert +\vert {\cal F}\vert )^{\Oh(1)}$, when every set is of size at most $d$ and 
has at most one red point. 
\end{enumerate}
Observe that the first algorithm generalizes the known algorithm for \PLC\ which runs in time  
$2^{\Oh(k_\ell \log k_\ell)}\cdot (\vert U\vert +\vert {\cal F}\vert)^{\Oh(1)}$~\cite{LP05}.  

The parameterized algorithm for \slrbsc\ mentioned in 
(a) starts by bounding the number of blue vertices by $k_\ell^2$ and guessing the lines that contain at least two blue points. The number of  lines containing at least two blue points can be shown to be at most  $k_\ell^4$. These guesses lead to an equivalent instance where each line contains exactly one blue point and there are no lines that only contain red points (as these lines can be deleted). However, we can not bound the number of red points at this stage. 
We introduce a notion of "solution subfamily" and connected components of the solution subfamilies. Interestingly, this equivalent instance has sufficient geometric structure on the connected components. 
We exploit the structure of these components, gotten mainly from simple properties of lines on a plane, to show that knowing one of the lines in each component can, in \FPT\ time, lead to finding the component itself!  Thus, to find a component all we need to do is to guess one of the lines in it. However, here  we face our second difficulty: the number of connected components can be as bad as $\Oh(k_\ell)$ and thus if we guess one line for each connected component then it would lead to a factor of 
$\vert {\cal F}\vert ^{\Oh(k_\ell)}$ in the running time of the algorithm. However, our equivalent instances are such that we are allowed to process each component independent of other components. This brings the total running time of guessing the first line of each component down to $k_\ell\cdot \vert {\cal F}\vert $. The algorithmic ideas used here can be viewed as some sort of ``geometry preserving subgraph isomorphism'', which could be useful in other contexts also. This completes an overview of the \FPT\ result for \slrbsc\ parameterized by $k_\ell+k_r$.

The algorithm for \srbsc\  running in time $2^{\Oh(dk_\ell)}\cdot (\vert U\vert +\vert {\cal F}\vert)^{\Oh(1)}$, where every set is of size at most $d$ and has at most one red point is purely based on a novel reduction to \GI\ where the subgraph we are looking for has size $\Oh(k_\ell d)$ and treewidth $3$. The host graph, where we are looking for a solution subgraph, is obtained by starting with the bipartite incidence graph and making modifications to it. The bipartite incidence graph we start with has in one side vertices for sets and in the other side vertices corresponding to blue and red points  and there is an edge between vertices corresponding to a set and a blue (red) point if this blue (red) point is contained in the set. Our main observation is that a solution subfamily can be captured by a subgraph of size $\Oh(k_\ell d)$ and treewidth $3$. Thus, for our algorithm we enumerate all such subgraphs in time $2^{\Oh(dk_\ell)}\cdot (\vert U\vert +\vert {\cal F}\vert )^{\Oh(1)}$ and for each such subgraph we check whether it exists in the host graph using known algorithms for \GI. This concludes the description of this algorithm.

%

%
\section{Preliminaries}
In this paper an undirected graph is denoted by a tuple $G = (V,E)$, where $V$ denotes the set of vertices and $E$ the set of edges. For a set $S\subseteq V$, the {\it subgraph of $G$ induced by $S$}, denoted by $G[S]$, is defined as the subgraph of $G$ with vertex set $S$ and edge set $\{(u,v) \in E :u,v\in S\}$. The subgraph obtained after deleting $S$ is denoted as $G\setminus S$. All vertices adjacent to a vertex $v$ are called neighbors of $v$ and the set of all such vertices is called the neighborhood of $v$. Similarly, a non-adjacent vertex of $v$ is called a non-neighbor and the set of all non-neighbors of $v$ is called the non-neighborhood of $v$. The neighborhood of $v$ is denoted by $N(v)$. A vertex in a connected graph is called a cut vertex if its deletion results in the graph becoming disconnected.



Recall that showing a problem \W[1] or \W[2] hard implies that the problem is unlikely to be \FPT.  
One can show that a problem is \W[1]-hard (\W[2]-hard) by presenting a parameterized reduction from a known 
\W[1]-hard problem (\W[2]-hard) such as {\sc Clique} (\SC) to it. The most important property of a parameterized reduction is that it corresponds to an \FPT\ algorithm that bounds the parameter value of the constructed instance by a function of the parameter of the source instance. A parameterized problem is said to be in the class para-\NP{} if it has a nondeterministic algorithm with \FPT\ running time. To show that a problem is para-\NP-hard we need to show that the problem is \NP-hard for some constant value of the parameter. For an example {\sc $3$-Coloring} is  para-\NP-hard parameterized by the number of colors. See~\cite{FlumGrohebook} for more details.

\medskip 

\noindent
{\bf Lower bounds in Kernelization.}
In the recent years, several techniques have been developed to show that certain parameterized problems belonging to the \FPT\ class cannot have any polynomial sized kernel unless some classical complexity assumptions are violated. One such technique that is widely used is the polynomial parameter transformation technique.

\begin{definition}
 Let $\Pi,\Gamma$ be two parameterized problems. A polynomial time algorithm $\mathcal A$ is called a polynomial parameter transformation (or ppt) from $\Pi$ to $\Gamma$ if , given an instance $(x,k)$ of $\Pi$, $\mathcal A$ outputs in polynomial time an instance $(x',k')$ of $\Gamma$ such that $(x,k) \in \Pi$ if and only if $(x',k')\in \Gamma$ and $k' \leq p(k)$ for a polynomial $p$.
\end{definition}

We use the following theorem together with ppt reductions to rule out polynomial kernels. 
\begin{theorem}
Let $\Pi,\Gamma$ be two parameterized problems such that $\Pi$ is \NP-hard and $\Gamma \in \NP$. Assume that there exists a polynomial parameter transformation from $\Pi$ to $\Gamma$. Then, if $\Pi$ does not admit a polynomial kernel neither does $\Gamma$.
\end{theorem}
For further details on lower bound techniques in kernelization refer to \cite{BodlaenderDFH09,FortnowS11}. 

\medskip

\noindent
{\bf  Generalized Red Blue Set Cover.} 
A set $S$ in an \RBSC\ instance $(U,\F)$ is said to {\em cover}  a point $p \in U$ if $p \in S$.
A {\em solution family} for the instance is a family of sets of size at most $k_\ell$ that covers all the blue points and at most $k_r$ red points. In case of \oRBSC, the solution family is simply a family of sets that covers all the blue points but at most $k_r$ red points. Such a family will also be referred to as a {\em valid family}. A {\em minimal family of sets} is a family of sets such that every set contains a unique blue point. In other words, deleting any set from the family implies that a strictly smaller set of blue points is covered by the remaining sets.  The sets of \LRBSC\ are also called {\em lines} in this paper. We also mention a key observation about lines in this section. This observation is crucial in many arguments in this paper.

\begin{observation}\label{covering_lines_bd}
 Given a set of points $S$, let $\F$ be the set of lines such that each line contains at least $2$ points from $S$. Then $\vert \F\vert \leq {\vert S\vert \choose 2}$. 
\end{observation}

\srbsc\ with hyperplanes of ${\mathbb R}^d$, for a fixed positive integer $d$, is a special case for the problem. Here, the input universe $U$ is a set of $n$ points in ${\mathbb R}^d$. A hyperplane in ${\mathbb R}^d$ is the affine hull of a set of $d+1$ affinely independent points \cite{LP05}. In our special case each set is a maximal set of points that lie on a hyperplane of ${\mathbb R}^d$.


\begin{definition}
 An intersection graph $G_{\F}=(V,E)$ for an instance $(U,\F)$ of \RBSC\ is a graph with vertices corresponding to the sets in $\F$. We give an edge between two vertices if the corresponding sets have non-empty intersection.
\end{definition}

The following proposition is a collection of results on the \SC\ problem, that will be repeatedly used in the paper. The results are from \cite{DLS2014,DF99}
\begin{prop}\label{SC_results}
The \SC\ problem is:
\begin{enumerate}[(i)]
\item \W[2] hard when parameterized by the solution family size $k$.
\item \FPT\ when parameterized by the universe size $n$, but does not admit polynomial kernels unless $\CoNP \subseteq \NP/\mbox{poly}$.
\item \FPT\ when parameterized by the number of sets $m$ in the instance, but does not admit polynomial kernels unless $\CoNP \subseteq \NP/\mbox{poly}$.
\end{enumerate}
\end{prop} 

\noindent
{\bf Tree decompositions and treewidth.} We also need the concept of treewidth and tree decompositions.  
\begin{definition}[Tree Decomposition ~\cite{Robertson1984}] A tree decomposition of a (undirected or directed) graph $G=(V,E)$ is a tree
 $\mathbb{T}$ in which each vertex $x \in \mathbb{T}$ has an assigned set of vertices $B_x \subseteq V$ (called a bag) such 
that $(\mathbb{T},\{B_x\}_{x \in \mathbb{T}})$
has the following properties:
\begin{itemize}
\item $\bigcup_{x\in \mathbb{T}} B_x = V$
 \item For any $(u,v) \in E$, there exists an $x \in \mathbb{T}$ such that $u, v \in B_x$.
\item If $v \in B_x$ and $v \in B_y$, then $v \in B_z$ for all $z$ on the path from $x$ to $y$ in $\mathbb{T}$.
\end{itemize}
In short, we denote $(\mathbb{T},\{B_x\}_{x \in \mathbb{T}})$ as $\mathbb{T}$.
\end{definition}

The \emph{treewidth} $tw(\mathbb{T})$ of a tree decomposition $\mathbb{T}$ is the size of the largest bag of $\mathbb{T}$ minus one. A graph may have several distinct tree decompositions. The treewidth $tw(G)$ of
a graph $G$ is defined as the minimum of treewidths over all possible tree decompositions of $G$.

\section{Parameterizing by  $k_r$ and $r$ }

In this section we first show that \slrbsc\ parameterized by $r$ is para-\NP-complete. Since $k_r \leq r$, it follows that 
\slrbsc\ parameterized by $k_r$ is also para-\NP-complete.

\begin{theorem}\label{thm:whardrkr}
\slrbsc\ is para-\NP-complete parameterized by either $r$ or $k_r$.
\end{theorem}
\begin{proof}
 If we are given a solution family for an instance of \slrbsc\ we can check in polynomial time if it is valid. 
 Hence, \slrbsc\  has a nondeterministic algorithm with \FPT\ running time (in fact polynomial) and thus  \slrbsc\  parameterized by  $r$ is in para-\NP. 

For completeness, there is an easy polynomial-time many-one reduction from the \PLC\ problem, which is \NP-complete. An instance $((U,\F))$ of \PLC\ parameterized by $k$, the size of the solution family, is reduced to an instance $((R \cup B,\F))$ of \slrbsc\ parameterized by $r$ or $k_r$ with the following properties: 
 \begin{itemize}
  \item $B = U$ 
  \item The family of sets remains the same in both instances.
  \item $R$ consists of $1$ red vertex that does not belong to any of the lines of $\F$.
  \item $k_\ell = k$ and $k_r =0$.
 \end{itemize}
It is easy to see that $((U,\F))$ is a \YES\ instance of \PLC\ if and only if $(R \cup B,\F)$ is a \YES\ instance of \slrbsc. Since the reduced instances belong to \slrbsc\ parameterized by $r=1$ or $k_r =0$, this proves that \slrbsc\ parameterized by $r$ or $k_r$ is para-\NP-complete.
\end{proof}


\section{Parameterizing by $\ell$}
In this section we design a parameterized algorithm as well as a kernel for \slrbsc\ when parameterized by the size $\ell$ of the family. The algorithm for this is simple. We enumerate all possible $k_\ell$-sized subsets of input lines and for each subset, we check in polynomial time whether it covers all blue points and at most $k_r$ red points. The algorithm runs in time 
$\Oh(2^\ell \cdot (\vert U\vert +\vert \F \vert)$. The main result of this section is a polynomial kernel for \slrbsc\ when parameterized by $\ell$. 

We start by a few reduction rules which will be used not only in the kernelization algorithm given below but also in other parameterized and kernelization algorithms in subsequent sections. 

\begin{Reduction Rule}\label{redn1}
 If there is a set \(S \in \F\) with only red points then delete $S$ from $\F$.
\end{Reduction Rule}

\begin{lem}
Reduction Rule~\ref{redn1} is safe. 
\end{lem}
\begin{proof}
 Let $\F'$ be a family of at most $k_\ell$ lines of the given instance that cover all blue points and at most $k_r$ red points. If $\F'$ contains $S$, then $\F' \setminus \{S\}$ is also a family of at most $k_\ell$ lines that cover all blue points and at most $k_r$ red points. Hence, we can safely delete $S$. This shows that Reduction Rule~\ref{redn1} is safe. 
\end{proof}


\begin{Reduction Rule}\label{red_vertices_bd}
 If there is a set \(S\in \F\) with more than \(k_r\) red points in it then delete \(S\) from \(\F\). 
\end{Reduction Rule} 

\begin{lem}
Reduction Rule~\ref{red_vertices_bd} is safe. 
\end{lem}
\begin{proof}
 If $S$ has more than $k_r$ red points then $S$ alone exceeds the budget given for the permissible number of covered red points. Hence, $S$ cannot be part of any solution family and can be safely deleted from the instance. This shows that Reduction Rule~\ref{red_vertices_bd} is safe. 
\end{proof}


Our final rule is as follows.  A similar Reduction Rule was used in~\cite{LP05}, for the \PLC\ problem. 

\begin{Reduction Rule}\label{redn4}
 If there is a set \(S \in \F\) with at least $k_\ell+1$ blue points then reduce the budget of $k_\ell$ by $1$ and the budget of $k_r$ by $|R \cap S|$. The new instance is  $(U \setminus S, \widetilde{{\cal F}})$,  where  
 $ \widetilde{{\cal F}} =\{F\setminus S~|~F\in {\cal F} \mbox{ and } F\neq S\}$.
%
\end{Reduction Rule}

\begin{lem}
Reduction Rule~\ref{redn4} is safe. 
\end{lem}
\begin{proof}
If $S$ is not part of the solution family then we need at least $k_\ell+1$ lines in the solution family to cover the blue points in $S$, which is not possible. Hence any solution family must contain $S$. 

Suppose the reduced instance has a solution family $\F'$ covering $B \setminus S$ blue points and at most $k_r-\vert R\cap S\vert$ red points from $R \setminus S$. Then $\F'\cup \{S\}$ is a solution for the original instance. On the other hand, suppose the original instance has a solution family $\hat{\F}$. As argued above, $S \in \hat{\F}$. $\hat{\F} \setminus S$ covers all blue points of $B \setminus S$ and at most $k_r - \vert R \cap S \vert$ red points from $R \setminus S$, and is a candidate solution family for the reduced instance. Thus, Reduction Rule~\ref{redn4} is safe.
\end{proof}


%
%
%
%

The following simple observation can be made after exhaustive application of Reduction Rule~\ref{redn4}. 

\begin{observation}\label{blue_size}
 If the budget for the subfamily $\F'$ to cover all blue and at most $k_r$ red points is $k_\ell$ then after exhaustive  applications of Reduction Rule~\ref{redn4} there can be at most $b\leq k_\ell^2$ blue points remaining in a \YES\ instance. If there are more than $k_\ell^2$ blue points remaining to be covered then we correctly say \NO.
\end{observation}

It is worth mentioning that even if we had weights on the red points in $R$ and asked for a solution family of size at most $k_\ell$ that covered all blue points but red points of weight at most $k_r$, then this weighted version, called \wslrbsc\ parameterized by $\ell$ is \FPT.  The \wslrbsc\ problem will be useful in the theorem below. 
Finally, we get the following result. 

\begin{thmk}\label{lines_kernel}
There is an algorithm for \slrbsc\  running in time $\Oh(2^\ell \cdot(\vert U \vert + \vert \F \vert))$. In fact, \slrbsc\  admits a polynomial kernel 
parameterized by $\ell$. 
\end{thmk}

\begin{proof}
We have already described the enumeration based algorithm at the beginning of this section. Here, we only give the polynomial kernel. Given an instance of \slrbsc\ we exhaustively apply Reduction Rules~\ref{redn1},~\ref{red_vertices_bd} and~\ref{redn4} to obtain an equivalent instance. By Observation~\ref{blue_size} and the fact that $k_\ell \leq \ell$, the current instance must have at most $\ell^2$ blue points, or we can safely say \NO. Also, the number of red points that belong to $2$ or more lines is bounded by the number of intersection points of the $\ell$ lines, i.e., $\ell^2$. Any remaining red points belong to exactly $1$ line. We reduce our \slrbsc\ instance to a \wslrbsc\ instance as follows:
\begin{itemize}
\item The family of lines and the set of blue points remain the same in the reduced instance. The red points appearing in the intersection of two lines also remain the same. Give a weight of $1$ to these red points.
\item For each line $L$, let $c(L)$ indicate the number of red points that belong exclusively to $L$. Remove all but one of these red points and give weight $c(L)$ to the remaining exclusive red point. 
\end{itemize}

In the \wslrbsc\ instance, there are $\ell$ lines, at most $\ell^2$ blue points and at most $\ell^2 + \ell$ red points. 
For each line $L$, the value of $c(L)$ is at most $k_r$, after Reduction Rule~\ref{red_vertices_bd}. Suppose $k_r > 2^\ell$. Then $r > 2^\ell$ and the parameterized algorithm for   \slrbsc\  running in time $\Oh(2^\ell \cdot(\vert U \vert +\vert \F \vert))$
runs in polynomial time. Thus we can assume that $k_r \leq 2^\ell$. Then we can represent $k_r$ and therefore the weights $c(L)$ by at most $\ell$ bits. Thus, the reduced instance has size bounded by $\Oh(\ell^2)$. 

Observe that we got an instance of  \wslrbsc\ and not of \slrbsc\, which is the requirement for the kernelization procedure. 
All this shows is that the reduction is a ``compression'' from \slrbsc\ parameterized by $\ell$ to \wslrbsc\ parameterized by $\ell$. 
This is rectified as follows.  
Since both the problems belong to \NP, there is a polynomial time many-one reduction from  \wslrbsc\ to \slrbsc. Finally, using this polynomial time reduction, we obtain a polynomial size kernel for \slrbsc\ parameterized by $\ell$. 
 \end{proof}

Observe that the algorithm referred to in Theorem~\ref{lines_kernel} does not use the fact that sets are lines and thus it also works for \srbsc\ parameterized by $\ell$. However, it  follows from Proposition~\ref{SC_results}(iii) that \srbsc\ parameterized by $\ell$ does not admit a polynomial kernel.

\section{Parameterizing by  $k_\ell$, $b$ and $k_\ell+b$}
In this section we look at \slrbsc\ parameterized by $k_\ell$, $b$, and $k_\ell+b$.  There is an interesting connection between $b$ and $k_\ell$. As we are looking for minimal solution families, we can alway assume that $b \geq k_\ell$. On the other hand, Reduction Rule~\ref{redn4} showed us that for all practical purposes $b \leq k_\ell^2$.  Thus, in the realm of parameterized complexity $k_\ell$, $b$ and $k_\ell+b$ are the {\em same parameters}. That is, \slrbsc\ is \FPT\ parameterized by $k_\ell$ if and only if it is \FPT\ parameterized by $b$ if and only if it is \FPT\ parameterized by $k_\ell+b$. The same holds in the context of kernelization complexity. 
%
%
First, we show that \slrbsc\ parameterised by $k_\ell$ or $b$ is \W[1]-hard. Then we look at some special cases that turn out to be \FPT.

\subsection{Parameter $k_\ell + b$}
We look at \slrbsc\ parameterized by $k_\ell + b$. This problem is not expected to have a \FPT\ algorithm as it is \W[1]-hard.
We give a reduction to this problem from the \MCC\ problem, which is known to be \W[1] hard even on regular graphs \cite{MathiesonS08}.
\medskip 

 \defparproblem{{\MCC}}{A graph $G=(V,E)$ where $V = V_1 \uplus V_2 \uplus \ldots \uplus V_k$ with $V_i$ being an independent set for all $1 \leq i \leq k$, and an integer $k$.}{$k$}{Is there a clique $C \subseteq G$ of size $k$ such that $\forall 1 \leq i \leq k, C \cap V_i \neq \emptyset$.} \vspace{10 pt}

The clique containing one vertex from each part is called a {\em multi-colored clique}. 

\begin{theorem}
\label{thm:hardklb}
 \slrbsc\ parameterized by $k_\ell$ or $b$ or $k_\ell + b$ is \W[1]-hard.
\end{theorem}

\begin{proof}
We will give a reduction from \MCC\ on regular graphs. 
  Let $(G=(V,E),k)$ be an instance of \MCC, where $G$ is a $d$-regular graph. We construct an instance of \slrbsc\, $(R \cup B, \mathcal{F})$, as follows. Let $V = V_1 \uplus V_2 \uplus \ldots \uplus V_k$.
\begin{enumerate}

\item For each vertex class $V_i,1\leq i\leq k$, add two blue points $b_i$ at $(0,i)$ and $b'_i$ at $(i,0)$.
\item Informally, for each vertex class $V_i,1\leq i\leq k$ we do as follows. 
Let $L_k$ be the line that is parallel to $y$ axis and passes through the point $(k,0)$.  
Suppose there are $n_i$ vertices in $V_i$. We select $n_i$ distinct points, say ${\cal P}$,  in  ${\mathbb R}^2$ 
on the line $L$, such that if $(a_i,a_2)\in \cal P$ then  $a_i=k$ (as these are points on $L_k$) and $a_2$ lies in the interval $(i-1,i-\frac{1}{2})$. Now  for every point $p\in \cal P$ we draw the unique line between $(0,i)$ and the point $p$.  
Finally, we assign each line to a unique vertex in $V_i$. Formally, we do as follows. 
For each vertex class $V_i,1\leq i\leq k$ and each vertex $u \in V_i$, we choose a point $p^1_u \in {\mathbb R}^2$ with coordinates $(k,y_u)$, $i-1< y_u < i-\frac{1}{2}$. Also, for each pair $u\neq v \ \in V_i$, $y_u \neq y_v$. For each $u\in V_i$, we add the line $l^1_u$, defined by $b_i$ and $p^1_u$, to $\mathcal F$.  We call these {\em near-horizontal lines}. 
Observe that all the near-horizontal lines corresponding to vertices in $V_i$ intersect at $b_i$. Furthermore, for any two vertices 
$u\in V_i$ and $v \in V_j$, with $i\neq j$, the lines $l^1_u$ and $l^1_v$ do not intersect on a point with $x$-coordinate from the closed interval $[0,k]$. 

\item Similarly, for each vertex class $V_i,1\leq i\leq k$ and each vertex $u \in V_i$, we choose a point $p^2_u \in {\mathbb R}^2$ with coordinates $(x_u,k), i-1< x_u < i-\frac{1}{2}$. Again, for each pair $u\neq v \in V_i$, $y_u \neq y_v$. For each $u\in V_i$, we add the line $l^2_u$, defined by $b_i$ and $p^2_u$, to $\mathcal F$. Notice that for any $u,v \in V$, $l^1_u$ and $l^2_v$ have a non-empty intersection. We call these  {\em near-vertical lines}. Observe that all the near-vertical lines corresponding to vertices in $V_i$ intersect at $b_i'$. Furthermore, for any two vertices 
$u\in V_i$ and $v \in V_j$, with $i\neq j$, the lines $l^2_u$ and $l^2_v$ do not intersect on a point with $y$-coordinate from the closed interval $[0,k]$. However, a near-\ line and a near-vertical line will intersect at a point with both $x$ and $y$-coordinate from the closed interval $[0,k]$. The construction ensures that no $3$ lines in $\F$ have a common intersection.

\item For each edge $e=(u,v) \in E$, add two red points, $r_{uv}$ at the intersection of lines $l^1_u$ and $l^2_v$, and $r_{vu}$ at the intersection of lines $l^1_v$ and $l^2_u$.

\item For each vertex $v \in V$, add a red point at the intersection of the lines $l^1_v$ and $l^2_v$.
\end{enumerate}
This concludes the description of the reduced instance. Thus we have an instance $(R \cup B, \mathcal{F})$ of \slrbsc\ with $2n$ lines, $2k$ blue points and $2m+n$ red points. 
 \begin{claim}
  $G=(V,E)$ has a multi-colored clique of size $k$ if and only if $(R \cup B, \mathcal{F})$ has a solution family of $2k$ lines, covering the $2k$ blue points and at most $2(d+1)k - k^2$ red points.
 \end{claim}
 \begin{proof}
  Assume there exists a multi-colored clique $C$ of size $k$ in $G$. Select the $2k$ lines corresponding to the vertices in the clique. That is, select the subset of lines ${\cal F}'=\{l^j_u~|~ 1\leq j \leq 2, u \in C\}$ in the \slrbsc\ instance. Since the clique is multi-colored, these lines cover all the blue points. Each line (near-horizontal or near-vertical) has exactly $d+1$ red points. 
  Thus, the number of red points covered by ${\cal F}'$ is at most $(d+1)2k$. However, each red point corresponding to vertices in $C$ and the two red points corresponding to each edge in $C$ are counted twice. Thus, the number of red points covered by  ${\cal F}'$ is at most $(d+1)2k-k-2{k \choose 2}=2(d+1)k - k^2$. This completes the proof in the forward direction. 
%
  
  Now, assume there is a minimal solution family of size at most $2k$, containing at most $2(d+1)k - k^2$ red points. As no two blue points are on the same line and there are $2k$ blue points, there exists a unique line covering each blue point. Let ${\mathcal L}^1$ and ${\mathcal L}^2$ represent the sets of near-horizontal and near-vertical lines respectively in the solution family.  Observe that ${\mathcal L}^1$ covers $\{b_1,\ldots,b_k\}$ and ${\mathcal L}^2$ covers $\{b_1',\ldots,b_k'\}$. 
  Let $C=\{v_1,\ldots,v_k\}$ be the set of vertices in $G$ corresponding to the lines in  ${\mathcal L}^1$. We claim that $C$ forms a 
  multicolored $k$-clique in $G$. Since $b_i$ can only be covered by lines corresponding to the vertices in $V_i$ and 
${\mathcal L}^1$ covers $\{b_1,\ldots,b_k\}$ we have that $C\cap V_i \neq \emptyset$. It remains to show that for every pair of vertices in $C$ there exists an edge between them in $G$.  Let $v_i$ denote the vertex in  $C\cap V_i$. 

Consider all the lines in ${\mathcal L}^1$. Each of these lines are near-horizontal and have exactly $d+1$ red points. Furthermore, no two of them intersect at a red point. Since the total number of red points covered by ${\mathcal L}^1 \cup {\mathcal L}^2$ is at most $2(d+1)k - k^2$, we have that the $k$ lines in ${\mathcal L}^2$ can only cover at most $k(d+1)-k^2$ red points that are not covered by the lines in  ${\mathcal L}^1$. That is, the $k$ lines in ${\mathcal L}^2$ contribute at most $k(d+1)-k^2$ {\em new red points} to the solution. Thus, the number of red points that are covered by both ${\mathcal L}^1$ and   ${\mathcal L}^2$ is $k^2$.  Therefore, any two lines $l_1$ and $l_2$ such that $l_1 \in {\mathcal L}^1$ and $l_2 \in {\mathcal L}^2$ must intersect at a red point. This implies that either $l_1$ and $l_2$ correspond to the same vertex in $V$ or there exists an edge between the vertices corresponding to them. Let $C'=\{w_1,\ldots,w_k\}$ be the set of vertices in $G$ corresponding to the lines in  ${\mathcal L}^2$. Since $b_i'$ can only be covered by lines corresponding to the vertices in $V_i$ and ${\mathcal L}^2$ covers $\{b_1',\ldots,b_k'\}$ we have that $C'\cap V_i \neq \emptyset$. Let $w_i$ denote the vertex in $V_i$ such that $l^2_{w_i}\in {\mathcal L}^2 $ covers $b_i'$.  We know that $l^1_{v_i} $ and $l^2_{w_i} $ must intersect on a red point. However, by construction no two distinct vertices $v_i$ and $w_i$ belonging to the same vertex class $V_i$ intersect at red point. Thus $v_i=w_i$. This means $C=C'$. This, together with the fact that two lines $l_1$ and $l_2$ such that 
$l_1 \in {\mathcal L}^1$ and $l_2 \in {\mathcal L}^2$ (now lines corresponding to $C$) must intersect at a red point, implies that 
$C$ is a multicolored $k$-clique in $G$. 
 \end{proof}
 Since $b= k_\ell=2k$, we have that \slrbsc\ is  \W[1]-hard parameterized by $k_\ell$ or $b$ or $k_\ell + b$. This concludes the proof. 
\end{proof}


A closer look at the reduction shows that every set contains exactly one blue point. A natural question to ask is whether the complexity would change if we take the complement of this scenario, that is, each set contains either no blue points or at least two blue points. Shortly, we will see that this implies that the problem becomes \FPT. Also, notice that each set in the reduction contains unbounded number of red elements. What about the parameterized complexity if every set in the input contained at most a bounded number, say $d$, of red elements. Even then the complexity would change but for this we need an algorithm for \slrbsc\  parameterized by $k_\ell+k_r$ that will be presented in Section~\ref{redsoln_comb}.


\subsection{Special case under the  parameter $k_\ell$}

In this section, we look at the special case when every line in the \slrbsc\ instance contains at least $2$ blue points or no blue points at all. We show that in this restricted case \slrbsc\ is \FPT.

\begin{thmk}\label{branching}
\slrbsc\ parameterized by $k_\ell$, where input instances have each set containing either at least $2$ blue points or no blue points, has a polynomial kernel. There is also an \FPT\ algorithm running in $\Oh(k_\ell^{4k_\ell}\cdot (\vert U \vert + \vert {\cal F} \vert)^{\Oh(1)})$ time.
\end{thmk}

\begin{proof}
We exhaustively apply Reduction Rules~\ref{redn1},~\ref{red_vertices_bd} and~\ref{redn4} to our input instance. In the end, we obtain an equivalent instance that has at least $1$ blue point per line. The equivalent instance also has each line containing at least $2$ blue points or no blue points. The instance has at most $b= k_\ell^2$ blue points, or else we can correctly say \NO. By Observation~\ref{covering_lines_bd} and the assumption on the instance, we can bound $\ell$ by ${b \choose 2} \leq k_\ell^4$. Now from Theorem~\ref{lines_kernel} we get a polynomial kernel for this special case of \slrbsc\ parameterized by $k_\ell$. 

Regarding the \FPT\ algorithm,  we are allowed to choose at most $k_\ell$ solution lines from a total of $\ell \leq k_\ell^4$ lines in the instance (of course after we have applied Reduction Rules~\ref{redn1},~\ref{red_vertices_bd} and~\ref{redn4} exhaustively). For every possible $k_\ell$-sized set of lines we check whether the set covers all blue vertices and at most $k_r$ red vertices. If the instance is a \YES\ instance, one such $k_\ell$-sized set is a solution family. This algorithm runs in 
$\Oh( {k_\ell^4 \choose k_\ell}\cdot (\vert U \vert + \vert {\cal F} \vert )^{\Oh(1)})=\Oh(k_\ell^{4k_\ell}\cdot (\vert U\vert + \vert {\cal F} \vert)^{\Oh(1)})$ time.  
\end{proof}

\section{Parameterizing by  $k_r+k_\ell$ and $b+k_r$}\label{redsoln_comb}

In the previous sections we saw that \slrbsc\ parameterized by $r$ is para-\NP-complete and is  \W[1]-hard parameterized by 
$k_\ell$.  So there is no hope of an \FPT\ algorithm unless $P = \NP$ or \FPT\ =\W[1], when parameterized by $r$ and $k_\ell$ respectively.  As a consequence, we consider combining different natural parameters with $r$ to see if this helps to find \FPT\ algorithms. In fact, in this section, we describe a \FPT\ algorithm for \slrbsc\ parameterized by $k_\ell+k_r$.  Since $k_r \leq r$, this implies that \slrbsc\ parameterized by $k_\ell+ r$ is \FPT. This is one of our main technical/algorithmic contribution. 
Also, since $k_\ell \leq b$ for any minimal solution family of an instance, it follows that \slrbsc\ parameterized by $b+k_r$ belongs to \FPT. It is natural to ask whether the \srbsc\ problem, that is, where sets in the family are arbitrary subsets  of the universe and need not correspond to lines, is \FPT\ parameterized by $k_\ell+k_r$. In fact, Theorem~\ref{MeCC_redn} states that the problem is \W[1]-hard even when each set is of size three and contains at least two red points. This shows that indeed restricting ourselves to sets corresponding to lines makes the problem tractable.

We start by considering a simpler case, where the input instance is such that every line contains exactly $1$ blue point. Later we will show how we can reduce our main problem to such instances. By the restrictions assumed on the input, no two blue points can be covered by the same line and any solution family must contain at least $b$ lines. Thus, $b\leq k_\ell$ or else, it is a \NO\ instance. Also, a minimal solution family will contain at most $b\leq k_\ell$ lines. Hence, from now on we are only interested in the existence of minimal solution families. In fact, inferring from the above observations, a minimal solution family, in this special case, contains exactly $b$ lines. Let $G_{\F'}$ be the intersection graph that corresponds to a minimal solution $\F'$. Recall, that in $G_{\F'}$ vertices correspond to lines in $\F'$ and there is an edge between two vertices  in $G_{\F'}$ if the corresponding lines intersect either at a blue point or a red point. Next, we define notions of 
{\em good tuple} and {\em conformity} which will be useful in designing the \FPT\ algorithm for the special case. 
Essentially, a good tuple provides a numerical representation of connected components of $G_{\F'}$. 

\begin{definition} \label{goodtuple}Given an instance $(R,B,{\cal F})$ of \slrbsc\, we call a tuple 
$$\Big(b,p,s, P,   \{I'_1,\ldots,I'_s\}, (k_r^1,k_r^2,\ldots, k_r^s)\Big)$$  {\em good} if the following hold. 
\begin{enumerate}[(a)] 
 \item Integers $p\leq k_r$  and $s\leq b \leq k_\ell$; Here $b$ is the number of blue vertices in the instance.
 \item $P=P_1 \cup \cdots \cup  P_s$ is an $s$-partition of $B$;
 \item For each $1 \leq i \leq s$, $I'_i$ is an ordering for the blue points in part $P_i$;
 \item Integers $k_r^i,1 \leq i \leq s$, are such that $\Sigma_{1 \leq i \leq s} k_r^i = p$. 
\end{enumerate}
\end{definition}

Below, we define the relevance of good tuples in the context of our problem.
\begin{definition}\label{conforming_tuple}
 We say that the minimal solution family $\F'$ {\em conforms} with a good tuple $\Big(b,p,s, P,   \{I'_1,\ldots,I'_s\}, (k_r^1,k_r^2,\ldots, k_r^s)\Big)$
 if the following properties hold:
 \begin{enumerate}
  \item The components $C_1,\ldots, C_s$ of $G_{\F'}$ give the partition $P = P_1,\ldots, P_s$ on the blue points.
    \item For each component $C_i$, $1 \leq i \leq s$, let $t_i= \vert P_i \vert$. Let $I_i'=b_1^i,\ldots,b_{t_i}^i$  be an ordering of blue points in $P_i$. Furthermore assume that $L^i_j \in \F'$ covers the blue point $b_j^i$. $I_i'$ has the property that for all $j\leq t_i$ $G_{\F'}[\{L_1^i,\ldots,L_{j}^i \} ]$ is connected. In other words for all $j\leq t_i$, $L_{j}^i$ intersects with at least one of the lines from the set $\{L_1^i,\ldots,L_{j-1}^i \}$.  Notice that, by minimality of  
     $\F'$, the point of intersection for such a pair of lines is a red point.  
    
  \item $\F'$ covers $p \leq k_r$ red points.
  \item In each component $C_i$, $k_r^i$ is the number of red points covered by the lines in that component. It follows that $\Sigma_{1 \leq i \leq s} k_r^i =p$. In other words, the integers $k_r^i$ form a combination of $p$.
 \end{enumerate}
\end{definition}

The next lemma says that the existence of a minimal solution subfamily $\F'$ results in a conforming good tuple. 
\begin{lem}
\label{lem:conformingoodtuplexists}
Let  $(U,\F )$ be an input to  \slrbsc\ parameterized by $k_\ell+ k_r$, such that every line contains exactly $1$ blue point. If there exists a solution subfamily $\F'$ then there is a conforming good tuple. 
\end{lem}
\begin{proof}
Let $\F'$ be a minimal solution family of size $b \leq k_\ell$ that covers $p\leq k_r$ red points. Let $G_{\F'}$ have $s$ components viz. $C_1,C_2,\cdots,C_{s}$, where $s \leq k_\ell$. For each $i \leq s$, let $\F_{C_i}$ denote the set of lines corresponding to the vertices of $C_i$. $P_i = B \cap \F_{C_i}$, $t_i = \vert P_i \vert$ and $k_r^i = \vert R \cap \F_{C_i} \vert$. In this special case and by minimality of $\F'$, $\vert \F_{C_i} \vert = t_i$. As $C_i$ is connected, there is a sequence $\{L_1^i,L_2^i,\ldots L_{t_i}^i\}$ for the lines in $\F_{C_i}$ such that for all $j\leq t_i$ we have that 
    $G_{\F'}[\{L_1^i,\ldots,L_{j}^i \} ]$ is connected. This means that, for all $j\leq t_i$ $L_{j}^i$ intersects with at least one of the lines from the set $\{L_1^i,\ldots,L_{j-1}^i \}$. By minimality of $\F'$, the point of intersection for such a pair of lines is a red point. For all $j \leq t_i$, let $L_j^i$ cover the blue point $b_j^i$. Let $I'_i = b_1^i,b_2^i,\ldots,b_{t_i}^i$. 
The tuple $\Big(b,p,s, P=P_1\cup P_2 \ldots \cup P_s, \{I'_1,\ldots,I'_s\}, (k_r^1,k_r^2,\ldots, k_r^s)\Big)$ is a good tuple and it also conforms with $\F'$.   This completes the proof.  
\end{proof}


The idea of the algorithm is to generate all good tuples and then check whether there is a 
solution subfamily $\F'$ that conforms to it. The next lemma states we can check for a conforming minimal solution family when we are given a good tuple. 
\begin{lem}\label{good_tuples}
 For a good tuple $(b,p,s, P,   \{I'_1,\ldots,I'_s\}, (k_r^1,k_r^2,\ldots, k_r^s))$, 
 we can verify in $\Oh(b \ell p^b)$ time whether there is a minimal solution family $\F'$ that conforms with this tuple.
\end{lem}

\begin{proof}
The algorithm essentially builds a search tree for each partition $P_i, 1\leq i \leq s$. For each part $P_i$, we define a set of points $R'_i$ which is initially an empty set. 


For each $1 \leq i \leq s$, let $t_i=|P_i|$ and let $I_i'=b_1^i,\ldots,b_{t_i}^i$  be the ordering of blue points in $P_i$. Our objective is to check whether there is a subfamily $\F_i' \subseteq \F$ such that it covers $b_1^i,\ldots,b_{t_i}^i$, and at most $k_r^i$ red point. At any stage of the algorithm, we have a subfamily $\F_i'$ covering  
$b_1^i,\ldots, b_j^i$ and at most $k_r^i$ red points. In the next step we try to enlarge $\F_i'$  in such a way 
that it also covers $b_{j+1}^i$, but still covers at most $k_r^i$ red points. In some sense we follow the ordering given by $I_i'$ to build $\F_i'$. 

Initially, 
$\F_i'=\emptyset$. At any point of the recursive algorithm we represent the problem to be solved by the following tuple: 
($\F_i'$, $R_i'$, ($b_j^i,\ldots,b_{t_i}^i$), $k_r^i -|R_i'|$). 
We start the process by guessing the line in $\cal F$ that covers $b_1^i$, say $L_1^i$. That is, for every $L\in \F$ such that 
$b_1^i$ is contained in $L$ we recursively check whether there is a solution to the tuple 
 ($\F_i':=\F_i'\cup \{L\}$, $R_i':=R_i'\cup (R \cap L) $, ($b_2^i,\ldots,b_{t_i}^i$),$k_r^i:=k_r^i -|R_i'|$).  If any tuple returns \YES\ then we return that there is a subset $\F_i' \subseteq \F$ which covers $b_1^i,\ldots,b_{t_i}^i$, and at most $k_r^i$ red points. 

Now suppose we are at an intermediate stage of the algorithm and the tuple we have is ($\F_i'$, $R_i'$, ($b_j^i,\ldots,b_{t_i}^i$), $k_r^i$). Let $\cal L$ be the set of lines such that it contains $b_j^i$ and a red point from $R_i'$. Clearly, $|{\cal L}|\leq |R_i'| \leq k_r^i$. For every line $L\in  {\cal L}$, we  recursively check whether there is a solution to the tuple ($\F_i':=\F_i'\cup \{L\}$, $R_i':=R_i'\cup (R \cap L) $, ($b_{j+1}^i,\ldots,b_{t_i}^i$),$k_r^i:=k_r^i -|R_i'|$).  If any tuple returns \YES\ then we return that there is a subset $\F_i' \subseteq \F$ which covers $b_1^i,\ldots,b_{t_i}^i$, and at most $k_r^i$ red points. 

Let $\mu=t_i$.  At each stage $\mu$ drops by one and, except for the first step, the algorithm recursively solves at most $k_r^i$ subproblems. This implies that the algorithm takes at most $\Oh(|\F| k_r^{t_i})= \Oh(\ell k_r^{t_i})$ time.

%

Notice that the lines in the input instance are partitioned according to the blue points contained in it. Hence, the search corresponding to each part $P_i$ is independent of those in other parts. In effect, we are searching for the components for $G_{\F'}$ in the input instance, in parallel. If for each $P_i$ we are successful in finding a minimal set of lines covering exactly the blue points of $P_i$ while covering at most $k_r^i$ red points, we conclude that a solution family $\F'$ that conforms to the given tuple exists and hence the input instance is a \YES\ instance.

The time taken for the described procedure in each part is at most $\Oh(\ell k_r^{t_i})$. Hence, the total time taken to check if there is a conforming minimal solution family $\F'$ is at most 
$$ \Oh(\ell \cdot \sum_{i=1}^s k_r^{t_i})=  \Oh(s \ell   p^b)=\Oh(b \ell p^b).$$ 
This concludes the proof. 
\end{proof}

We are ready to describe our \FPT\ algorithm for this special case of \slrbsc\ parameterized by $k_\ell+k_r$.
\begin{lem}\label{solution_1blue}
 Let  $(U,\F , k_\ell ,k_r)$ be an input to  \slrbsc\ such that every line contains exactly $1$ blue point. Then we can check whether there is a solution subfamily $\F'$ to this instance in time 
 $k_\ell^{\Oh(k_\ell)}\cdot k_r^{\Oh(k_r)}\cdot (\vert U \vert+ \vert {\cal F}\vert)^{\Oh(1)}$ time.
\end{lem}
\begin{proof}
Lemma~\ref{lem:conformingoodtuplexists} implies that for the algorithm all we need to do is to enumerate all possible good tuples 
$(b,p,s, P,   \{I'_1,\ldots,I'_s\}, (k_r^1,k_r^2,\ldots, k_r^s))$, and for each tuple, check whether there 
is a conforming minimal solution family. Later, we use the algorithm described in Lemma~\ref{good_tuples}. 
%
We first give an upper bound on the number of tuples  and how to enumerate them. 
\begin{enumerate}
\item There are $k_\ell$ choices for $s$ and $k_r$ choices for $p$.
\item There can be at most $b^{k_\ell}$ choices for $P$ which can be enumerated in  $\Oh(b^{k_\ell}\cdot k_\ell)$ time. 
\item For each $j\leq s$, $I'_j$ is ordering for blue points in $P_i$. Thus, if $|P_i|=t_i$, then the number of  ordering 
tuples $\{I'_1,\ldots,I'_s\}$ is upper bounded by $\prod_{i=1}^s t_i! \leq \prod_{i=1}^s t_i^{t_i} \leq  \prod_{i=1}^s b^{t_i} =b^b$. 
Such orderings can be enumerated in $\Oh(b^b)$ time.
 
\item For a fixed $p\leq k_r,s\leq k_\ell$, there are 
at most $p + s-1 \choose s -1$ solutions for $k_r^1+k_r^2+ \ldots + k_r^s = p$  
 and this set of solutions can be enumerated in $\Oh({{p+s-1} \choose {s-1}}\cdot ps)$ time. Notice that if $p \geq s$ then the time required for enumeration is $\Oh((2p)^{p}\cdot ps)$. Otherwise, the required time is $\Oh((2s)^{s}\cdot ps)$. As $p \leq k_r$ and $s \leq k_\ell$, the time required to enumerate the set of solutions is $\Oh(k_\ell^{\Oh(k_\ell)}k_r^{\Oh(k_r)}\cdot k_\ell k_r)$. 
\end{enumerate}
   Thus we can generate the set of tuples in time $k_\ell^{\Oh(k_\ell)}\cdot k_r^{\Oh(k_r)}$.Using Lemma~\ref{good_tuples}, for each tuple we check in at most $\Oh(k_r^{k_\ell}\cdot k_\ell \ell)$ time whether there is a conforming solution family or not. If there is no tuple with a conforming solution family, we know that the input instance is a \NO\ instance. 
   The total time for this algorithm is $k_\ell^{\Oh(k_\ell)}k_r^{\Oh(k_r)}k_r^{\Oh(k_\ell)}\cdot (\vert U \vert+\vert {\cal F} \vert)^{\Oh(1)}$. Again, if $k_r \leq k_l$ then $k_r^{\Oh(k_\ell)} = k_\ell^{\Oh(k_\ell)}$. Otherwise, $k_r^{\Oh(k_\ell)} = k_r^{\Oh(k_r)}$. Either way, it is always true that $k_r^{\Oh(k_\ell)} = k_\ell^{\Oh(k_\ell)}k_r^{\Oh(k_r)}$. Thus, we can simply state the running time to be 
   $k_\ell^{\Oh(k_\ell)}\cdot k_r^{\Oh(k_r)}\cdot (\vert U \vert + \vert {\cal F}\vert)^{\Oh(1)}$.  
\end{proof}


We return to the general problem of \slrbsc\ parameterized by $k_\ell+k_r$. Instances in this problem may have lines containing $2$ or more blue points. 
We use the results and observations described above to arrive at an \FPT\ algorithm for \slrbsc\ parameterized by $k_\ell+k_r$.

\begin{theorem}
\label{thm:colour_coding}
 \slrbsc\ parameterized by $k_\ell+k_r$ is \FPT, with an algorithm that runs in $k_\ell^{\Oh(k_\ell)}\cdot 
 k_r^{\Oh(k_r)} \cdot (\vert U \vert + \vert {\cal F} \vert)^{\Oh(1)}$ time.
\end{theorem}
\begin{proof}
Given an input $(U,\F , k_\ell ,k_r)$ for \slrbsc\ parameterized by $k_\ell+k_r$, we do some preprocessing to make the instance simpler. We exhaustively apply Reduction Rules~\ref{redn1}, \ref{red_vertices_bd} and~\ref{redn4}. After this, by Observation~\ref{blue_size}, the reduced equivalent instance has at most $k_\ell \choose 2$ blue points  if it is a \YES\ instance.

A minimal solution family can be broken down into two parts:  the set of lines containing at least $2$ blue points, and the remaining set of lines which contain exactly $1$ blue point.  Let us call these sets $\F_2$ and $\F_1$ respectively. We start with the following observation. 

\begin{observation}\label{2blue_lines}
 Let $\F'' \subseteq \F$ be the set of lines that contain at least $2$ blue points. There are at most $k_\ell^{4} \choose k_\ell$ ways in which a solution family can intersect with $\F''$.
\end{observation}
\begin{proof}
 Since $b \leq {k_\ell \choose 2}$, it follows from Observation~\ref{covering_lines_bd} that $\vert \F''\vert \leq k_\ell^4$. For any solution family, there can be at most $k_\ell$ lines containing at least $2$ blue points. Since the number of subsets of $\F''$ of size at most $k_\ell$ is bounded by $k_\ell^{4k_\ell}$, the observation is true.
\end{proof}

 From Observation~\ref{2blue_lines}, there are $k_\ell^{4k_\ell}$ choices for the set of lines in $\F_2$. We branch on all these choices of $\F_2$. On each branch, we reduce the budget of $k_\ell$ by the number of lines in $\F_2$ and the budget of $k_r$ by $\vert R \cap \F_2 \vert$. Also, we make some modifications on the input instance: we delete all other lines containing at least $2$ blue points from the input instance. We delete all points of $U$ covered by $\F_2$ and all lines passing through blue points covered by $\F_2$. Our modified input instance in this branch now satisfies the assumption of Lemma~\ref{solution_1blue} and we can find out in $k_\ell^{\Oh(k_\ell)}k_r^{\Oh(k_r)}\cdot (\vert U\vert +\vert {\cal F} \vert)^{\Oh(1)}$ time whether there is a minimal solution family $\F_1$ for this reduced instance. If there is, then $\F_2 \cup \F_1$ is a minimal solution for our original input instance and we correctly say \YES.
Thus the total running time of this algorithm is $k_\ell^{\Oh(k_\ell)} \cdot k_r^{\Oh(k_r)} \cdot (\vert U \vert + \vert {\cal F}\vert)^{\Oh(1)}$. 

It  may be noted here that for a special case where we can use any line in the plane as part of the solution, the second part of the algorithm becomes considerably simpler. Here for each blue point $b$, we can use an arbitrary line containing only $b$ and no red point.
 \end{proof}

\begin{corollary}\label{sp_case_red}
\slrbsc\ parameterized by $k_\ell+d$, where every line contains at most $d$ red points, is \FPT. The running time of the \FPT\ algorithm is $(dk_\ell)^{\Oh(dk_\ell)} \cdot (\vert U\vert +\vert {\cal F}\vert)^{\Oh(1)}$. The problem remains \FPT\ for all parameter sets $\Gamma'$ that contain $\{k_\ell,d\}$ or $\{b,d\}$. 
\end{corollary}

\begin{proof}
In this special case, any solution family can contain at most $dk_\ell$ red points. Hence we can safely assume that $k_r \leq dk_\ell$ and apply Theorem~\ref{thm:colour_coding}. 
\end{proof}
\subsection{Kernelization for \slrbsc\ parameterized by $k_\ell+k_r$ and $b+k_r$}

We give a polynomial parameter transformation from \SC\ parameterized by universe size $n$, to \slrbsc\ parameterized by $k_\ell+k_r+b$. Proposition \ref{SC_results}(ii) implies that on parameterizing by any subset of the parameters $\{k_\ell,k_r,b\}$, we will also obtain a negative result for polynomial kernels.

\begin{thmk}\label{set_cover_redn}
 \slrbsc\ parameterized by $k_\ell+k_r+b$ does not allow a polynomial kernel unless $\CoNP\subseteq \NP/\mbox{poly}$.
\end{thmk}
 
\begin{proof}
 
 Let $(U,\mathcal{S})$ be a given instance of \SC. Let $\vert U \vert =n, \vert {\mathcal S} \vert =m$. We construct an instance $(R \cup B,\F)$ of \slrbsc\ as follows. We assign a blue point $b_u \in B$ for each element $u \in U$ and a red point $r_S \in R$ for each set $S \in \mathcal{S}$. The red and blue points are placed such that no three points are collinear. We add a line between $b_u$ and $r_S$ if $u \in S$ in the \SC\ instance. Thus the \slrbsc\ instance $(R\cup B,\mathcal{F})$ that we have constructed has $b =n$, $r=m$ and $\ell =\sum_{S\in {\cal S}} |S|$. We set $k_r = k$ and $k_\ell=n$.
 
 \begin{claim}
  All the elements in $(U,\mathcal{S})$ can be covered by $k$ sets if and only if there exist $n$ lines in $(R\cup B,\mathcal{F})$ that contain all blue points but only $k$ red points.
 \end{claim}
\begin{proof}
 Suppose $(U,\mathcal{S})$ has a solution of size $k$, say $\{S_1,S_2,\cdots S_k\}$. The red points in the solution family for \slrbsc\ are $\{r_{S_1},r_{S_2},\cdots r_{S_k}\}$  corresponding to $\{S_1,S_2,\cdots S_k\}$. For each element $u \in U$, we arbitrarily assign a covering set $S_u$ from $\{S_1,S_2,\cdots S_k\}$. The solution family is the set of lines defined by the pairs 
 $\{(b_u, r_{S_u})~\mid ~ u \in U\}$. This covers all blue points.

 Conversely, if $(R\cup B,\mathcal{F})$ has a solution family $\F'$ covering $k$ red points and using at most $n$ lines, the sets in $\mathcal{S}$ corresponding to the red points in $\F'$ cover all the elements in $(U,\mathcal{S})$.
\end{proof}
If $k > n$, then the \SC\ instance is a trivial \YES\ instance. Hence, we can always assume that $k \leq n$.
This completes the proof that \slrbsc\ parameterized by $k_\ell+k_r+ b$ cannot have a polynomial sized kernel unless $\CoNP\subseteq \NP/\mbox{poly}$.
\end{proof}

\section{Hyperplanes: parameterized by $k_\ell+k_r$}
\begin{theorem}
 \srbsc\ for hyperplanes in ${\mathbb R}^d$, for a fixed positive integer $d$, is \W[1]-hard when parameterized by $k_\ell+k_r$.
\end{theorem}
\begin{proof}
 The proof of hardness follows from a reduction from k-CLIQUE problem. The proof follows a framework given in \cite{Mar06}.
 
 Let $(G(V,E),k)$ be an instance of k-CLIQUE problem. Our construction consists of a $k\times k$ matrix of gadgets $G_{ij}$, $1 \leq i,j, \leq k$. Consecutive gadgets in a row are connected by horizontal connectors  and consecutive gadgets in a column are connected by vertical connectors. Let us denote the horizontal connector connecting the gadgets $G_{ij}$ and $G_{ih}$ as $H_{i(jh)}$ and the vertical connector connecting the gadgets $G_{ij}$ and $G_{hj}$ as $V_{(ih)j}$, $1 \leq i,j,h \leq k$. 
 \\ \textbf{Gadgets}: The gadget $G_{ij}$ contains a blue point $b_{ij}$ and a set $R_{ij}$ of $d-2$ red points. In addition there are $n^2$ sets $R'_{ij}(a,b), 1 \leq a,b \leq n$, each having two red points each. 
 \\ \textbf{Connectors}: The horizontal connector $H_{i(jh)}$ has a blue point $b_{i(jh)}$ and a set $R_{i(jh)}$ of $d-2$ red points. Similarly, the vertical connector $V_{(ih)j}$ a blue point $b_{(ih)j)}$ and a set $R_{(ih)j}$ of $d-2$ red points.
 
 \noindent The points are arranged in general position i.e., no set of $d+2$ points lie on the same $d$-dimensional hyperplane. In other words, any set of $d+1$ points define a distinct hyperplane.
 \\ \textbf{Hyperplanes}: Assume $1\leq i,j,h \leq k$ and $1\leq a,b,c \leq n$. Let $P_{ij}(a,b)$ be the hyperplane defined by the $d+1$ points of $b_{ij} \cup R_{ij} \cup R'_{ij}(a,b)$. Let $P^h_{i(jh)}(a,b,c)$ be the hyperplane defined by $d+1$ points of $b_{i(jh)}\cup R_{i(jh)} \cup r_1 \cup r_2$ where $r_1 \in R'_{ij}(a,b)$ and $r_2 \in R'_{ih}(a,c)$. Let $P^v_{(ij)h}(a,b,c)$ be the hyperplane defined by $d+1$ points of $b_{(ij)h)}\cup R_{(ij)h} \cup r_1 \cup r_2$ where $r_1 \in R'_{ih}(a,c)$ and $r_2 \in R'_{jh}(b,c)$.
 \\For each edge $(a,b) \in E(G)$, we add $k(k-1)$ hyperplanes of the type $P_{ij}(a,b)$, $i \ne j$. Further, for all $1\leq a \leq n$, we add $k$ hyperplanes of the type $P_{ii}(a,a)$, $1\leq i \leq k$. The hyperplane $P^h_{i(jh)}(a,b,c)$ containing the blue point $b_{i(jh)}$ in a horizontal connector, is added to the construction if $P_{ij}(a,b)$ and $P_{ih}(a,c)$ are present in the construction. Similarly, the hyperplane $P^v_{(ij)h}(a,b,c)$ containing the blue point $b_{(ij)h}$ in a vertical connector, is added to the construction if $P_{ih}(a,c)$ and $P_{jh}(b,c)$ are present in the construction.

 Thus our construction has $k^2+2k(k-1)$ blue points, $(k^2+2k(k-1))(d-2) + 2n^2k^2$ red points and $O((m^2k^2)$ hyperplanes.
 \begin{claim}
  $G$ has a $k$-clique if and only if all the blue points in the constructed instance can be covered by $k^2+2k(k-1)$ hyperplanes covering at most $k^2d+2k(k-1)(d-2)$ red points.
 \end{claim}
\begin{proof}
Assume $G$ has a clique of size $k$ and let $\{a_1,a_2,\cdots,a_k\}$ be the vertices of the clique. Now we show a set cover of desired size exists. Choose $k$ hyperplanes, $P_{ii}(a_i,a_i), 1\leq i \leq k$, to cover the diagonal gadgets. To cover other gadgets,$G_{ij}$, choose the hyperplanes $P_{ij}(a_ia_j)$ and to cover the connectors, $H_{i(jh)}$ and $V_{(ih)j}$, choose the hyperplanes $P^h_{i(jh)}(a_i,a_j,a_h)$ and $P^v_{(ij)h}(a_i,a_j,a_h)$. The fact that $\{a_1,a_2,\cdots,a_k\}$ forms a clique implies that these hyperplanes do exist in the construction. 

Now assume a set cover of given size exists. 
To cover the blue point $b_{ij}$ in the gadget $G_{ij}$, any hyperplane adds $d$ red points. Also to cover the blue point in each connector, we need to add $d-2$ extra red points. Since each hyperplane contains $d$ red points and we have already used up our budget of red points, each hyperplane covering the connector points should reuse two red points that have been used in covering gadgets. By construction, this is possible  only when all gadgets in a row(column) are covered by hyperplanes corresponding to edges incident on the same vertex viz. the vertex corresponding to the hyperplane covering the diagonal gadget in the row(column). This implies that $G$ has a clique.
\end{proof}
\end{proof}

\section{Multivariate complexity of {\slrbsc}: Proof of Theorem~\ref{thm:main:rbscalg+kernel}}
The  first part of Theorem~\ref{thm:main:rbscalg+kernel} (parameterized complexity dichotomy) 
follows from Theorems~\ref{thm:whardrkr}, ~\ref{lines_kernel}, ~\ref{thm:hardklb} and \ref{thm:colour_coding}. 
Recall that $\Gamma=\{\ell,r,b,k_\ell,k_r\}$.  To show the kernelization dichotomy of the parameterizations of \slrbsc\ that admit \FPT\ kernels we do as follows: 
\begin{itemize}
\item Show that the problem admits a polynomial kernel parameterized by $\ell$ (Theorem~\ref{lines_kernel}). This implies that for all $\Gamma'$ that contains $\ell$, the parameterization admits a polynomial kernel.
\item Show that the problem does not admit a polynomial kernel when parameterized by $k_\ell+k_r+b$ (Theorem~\ref{set_cover_redn}). This implies that for all subsets of $\{k_\ell,k_r,b\}$, the parameterization does not allow a polynomial kernel. 
\item 
The remaining \FPT\ variants of \slrbsc\ correspond to parameter sets $\Gamma'$ that contain either $r$ or $\{r,b\}$ together. Recall that, $k_r\leq r$ and $k_\ell \leq b$. The two smallest combined parameters for which we can not infer the kernelization complexity from Theorem~\ref{set_cover_redn} are  $r+k_\ell$ and $r+b$. We show below (Theorem~\ref{thmk:rbscklr}) that \srbsc\  admits a quadratic kernel parameterized by   $r+k_\ell$. Since in any minimal solution family $k_\ell \leq b$, this also implies a quadratic kernel for the parameterization $r+b$. Thus, if parameterization by a set $\Gamma'$, which contains either $r$ or $\{r,b\}$, allows an \FPT\ algorithm then it also allows a polynomial kernel.   
\end{itemize}

\begin{thmk}
\label{thmk:rbscklr}
\slrbsc\ parameterized by  $k_\ell+r$   admits a polynomial kernel.
\end{thmk}
\begin{proof}
Given an instance of \slrbsc\ we first exhaustively apply Reduction Rules~\ref{redn1},~\ref{red_vertices_bd} and~\ref{redn4} and  obtain an equivalent instance. By Observation \ref{blue_size}, the reduced instance has at most $b\leq k_\ell^2$ blue points. By Observation~\ref{covering_lines_bd}, the number of lines containing at least two points is ${r+b}\choose{2}$. After applying Reduction Rule~\ref{redn1}, there are no lines with only one red point. Also, for a blue point $b_i$, if there are many lines that contain only $b_i$, then we can delete all but one of those lines. Therefore, the number of lines that contain exactly one point is bounded by $b$. Thus, we get a kernel of $k_\ell^2$ blue points, ${{r+k_\ell^2}\choose {2}}+k_\ell^2$ lines and $r$ red points. 
This concludes the proof.
%
\end{proof}

Combining Theorems~\ref{lines_kernel},~\ref{set_cover_redn} and ~\ref{thmk:rbscklr} and the discussion above we prove the second part of the Theorem~\ref{thm:main:rbscalg+kernel} (kernelization dichotomy).

\section{Parameterized Landscape for  {\sc \oLRBSC}}
Until now our main focus was the \slrbsc\ problem. In this section, we study the original \solrbsc\ problem. Recall that the original \solrbsc\ problem differs from the \slrbsc\ problem in the following way -- here our objective is {\em only} to minimize the number of red points that are contained in a solution subfamily, and {\em not} the size of the subfamily itself. 
That is, $k_\ell =\vert \F \vert$. This change results in a slightly different  
landscape for  \solrbsc\  compared to \slrbsc. As before let  $\Gamma=\{\ell,r,b,k_\ell,k_r\}$. 
We first observe that for all those $\Gamma'\subseteq \Gamma$ that do not contain $k_\ell$ as a parameter and \slrbsc\ is \FPT\ parameterized by $\Gamma'$, \solrbsc\  is also \FPT\ parameterized by $\Gamma'$. 
Next we list out the subsets of parameters for which the results do not follow from the result on  \slrbsc. 
\begin{itemize}
\item \solrbsc\ becomes \FPT\ parameterized by $r$.
\item \W[2]-hard parameterized by $k_r$. 
\end{itemize}

%

\subsection{\solrbsc\ parameterized by $r$}

\begin{thmk}\label{easy}
 \solrbsc\  parameterized by $r$ is \FPT. Furthermore,  \solrbsc\  parameterized by $r$ does not allow a polynomial kernel unless $\CoNP\subseteq \NP/\mbox{poly}$.  
\end{thmk}
\begin{proof}
We proceed by enumerating all possible $k_r$-sized subsets of $R$. For each subset, we can check in polynomial time whether the lines spanned by exactly those points cover all blue points. This is our \FPT\ algorithm, which runs in $\Oh(2^r \cdot (\vert U \vert + \vert {\cal F} \vert)^{\Oh(1)})$. 

Using Proposition~\ref{SC_results}, it is enough to show a polynomial parameter transformation from \SC\ parameterized by size $m$ of the set family, to \solrbsc\ parameterized by $r$. The reduction is exactly the same as the one given in the proof of Theorem~\ref{set_cover_redn}. This gives the desired second part of the theorem. 
\end{proof}

%
%

\subsection{\solrbsc\ parameterized by $k_r$}
In this section we study parameterization by $k_r$ and some special cases which leads to \FPT\ algorithm. 
We prove that \solrbsc\ parameterized by $k_r$ is \W[2]-hard. 
From Proposition~\ref{SC_results}, \SC\ parameterized by solution family size $k$ is \W[2]-hard. 
The \W[2]-hardness of \solrbsc\ parameterized by $k_r$ can be proved by a many-one reduction from 
\SC\ parameterized by $k$. The reduction is exactly the one that is given in Theorem~\ref{set_cover_redn}.

 \begin{theorem}\label{solution_red_hard}
 \solrbsc\ parameterized by  $k_r$ is \W[2]-hard.
\end{theorem}
 
\subsubsection{\FPT\ result under special assumptions}

In this section we consider a special case, where in the given instance every line contains either no red points or at least $2$ red points. There are two reasons motivating the study of this
special case. Firstly, in the \W[2]-hardness proof we crucially used the fact that the constructed \solrbsc\ instance has a set of lines with exactly $1$ red point. Thus, it is necessary to check if this 
is the reason leading to the hardness of the problem. Secondly, if we look at \sorbsc\ (sets in the family can be arbitrary) parameterized by $k_r$ and assumed that in the given instance every line contains either no red points or at least $2$ red points, then too the problem is \W[1]-hard (see Theorem~\ref{MeCC_redn}). However, when we consider \solrbsc\ parameterized by $k_r$ and where in the given instance every set contains either no red points or at least $2$ red points, the problem is \FPT. 


For our algorithm we also need the following new reduction rule. 
 \begin{Reduction Rule}\label{redn2}
  If there is a set \(S \in \F\) with only blue points then delete that set from $\F$
 and include the set in the solution.
 \end{Reduction Rule}
 
 \begin{lem}
 Reduction Rule~\ref{redn2} is safe. 
 \end{lem}
\begin{proof}
Since the parameter is $k_r$, there is no size restriction on the number of lines in the solution subfamily $\F'$. If $\F'$ is a solution subfamily and $S\in \F$ then under this parameterization, $\F' \cup \{S\}$ is also a solution family covering all blue points and at most $k_r$ red points. This shows that Reduction Rule~\ref{redn2} is valid.
\end{proof}

%
%
%
%
%

\begin{theorem}\label{cc_app}
 \solrbsc\ parameterized by $k_r$, where the input instance has every set containing at least $2$ red points or no red points at all, has an algorithm with running time $k_r^{\Oh(k_r^2)} \cdot (\vert U  \vert + \vert {\cal F} \vert )^{\Oh(1)}$.
\end{theorem}
\begin{proof}
Given an instance of \solrbsc, we first exhaustively apply Reduction Rules~\ref{redn1}, ~\ref{red_vertices_bd} and
 \ref{redn2}  and obtain an equivalent instance. At the end of these reductions we obtain an equivalent instance where every line has at least $1$ blue point and at least $2$ red points, but at most $k_r$ red points. 

Suppose $\F'$ is a solution family. Since a line with a red point has at least \(2\) red points, by Observation~\ref{covering_lines_bd}, the total number of sets that can contain the red points covered by $\F'$ is at most $k_r\choose 2$. This means that, if the input instance is a \YES\ instance, there exists a solution family with at most 
$k_\ell = {k_r\choose 2}$ lines. Now we can apply the algorithm for \slrbsc\ parameterized by $k_\ell+k_r$ described in Theorem~\ref{thm:colour_coding} to obtain an algorithm for \solrbsc\ parameterized by $k_r$. 
%
%
\end{proof}

Theorem~\ref{cc_app} gives an \FPT\ algorithm for  \solrbsc\ parameterized by $k_r$. In what follows we show that the same parameterization does not yield a polynomial kernel for this special case of \solrbsc. 
%
%
Towards this we give a polynomial parameter transformation from \SC\ parameterized by universe size $n$, to \solrbsc\ parameterized by $k_r$ and under the assumption that all sets in the input instance have at least $2$ red points.

\begin{thmk}
 \solrbsc\ parameterized by $k_r$, and under the assumption that all lines in the input have at least $2$ red points, does not allow a polynomial kernel unless $\CoNP\subseteq \NP/\mbox{poly}$.
\end{thmk}
 
\begin{proof}
 Let $(U,\mathcal{S})$ be a given instance of the \SC\ problem. We construct an instance $(R \cup B,\F)$ of \solrbsc\ as follows. We assign a blue point $b_u \in B$ for each element $u \in U$ and a red point $r_S \in R$ for each set $S \in \mathcal{S}$. The red and blue points are placed such that no three points are collinear. We add a line between $b_u$ and $r_S$ if $u \in S$ in the \SC\ instance. To every line $L$, defined by a blue point $b_u$ and a red points $r_S$, we add a unique red point $r_L \in R$. Thus the \solrbsc\ instance $(R\cup B,\mathcal{F})$ that we have constructed has $n$ blue points, 
 $\sum_{s\in {\cal S}}|S|$ lines and $m + \sum_{s\in {\cal S}}|S|$ red points. We set $k_r = k + n$.
 
 \begin{claim}
  All the elements in $(U,\mathcal{S})$ can be covered by $k$ sets if and only if there exist lines in $(R\cup B,\mathcal{F})$ that contain all blue points but only $k + n$ red points.
 \end{claim}
\begin{proof}
 Suppose $(U,\mathcal{S})$ has a solution of size $k$, say $\{S_1,S_2,\cdots S_k\}$. To each element $u\in U$, we arbitrarily associate a covering set $S_u$ from $\{S_1,S_2,\cdots S_k\}$. Our solution family $\F'$ of lines are the lines defined by the pairs of points $\{(b_u,r_{S_u})~|~ u \in U\}$. These lines cover all blue points. The number of red points contained in these lines are the $k$ red points $\{r_{S_1},r_{S_2},\cdots r_{S_k}\}$ associated with $\{S_1,S_2,\cdots S_k\}$, and the $n$ red points $\{r_L~|~ L \in \F'\}$. Therefore, in total there are $k+n$ red points in the solution.
 
Conversely, suppose $(R\cup B,\mathcal{F})$ has a family $\F'$ covering all blue points and at most $k + n$ red points. The construction ensures that at least $n$ lines are required to cover the $n$ blue points. This also implies that the unique red points belonging to each of these lines add to the number of red points contained in the solution family. The remaining $k$ red points, that are contained in the solution family, correspond to sets in $\mathcal{S}$ that cover all the elements in $(U,\mathcal{S})$.
\end{proof}
If $k > n$, then the \SC\ instance is a trivial \YES\ instance. Hence, we can always assume that $k \leq n$.
This completes the proof that \solrbsc\ parameterized by $k_r$, and under the assumption that every line in the input instance has at least $2$ red points, cannot have a polynomial sized kernel unless $\CoNP\subseteq \NP/\mbox{poly}$.
\end{proof}

\subsection{Proof of Theorem~\ref{thm:main:oldrbscalg+kernel}}
Proof of Theorem~\ref{thm:main:oldrbscalg+kernel} follows from Theorems~\ref{thm:main:rbscalg+kernel},~\ref{easy} and 
~\ref{solution_red_hard}.

\section{\sc Generalized Red Blue Set Cover}
In this section we show that for several parameterizations, under which \slrbsc\ is \FPT, the \srbsc\ problem is not. In this section we give the following three results which complement the corresponding results in the geometric setting. 

\begin{enumerate}
\item \srbsc\ is \W[1]-hard parameterized by $k_\ell+k_r$ when every set has size at most three and contains at least two red elements. 
\item \srbsc\ is \W[2]-hard parameterized by $k_\ell+r$ when every set contains at most one red element. 
\item  \srbsc\ is  \FPT, parameterized by $k_\ell$ and $d$, when every set has at most one red element.  Here, $d$ is 
the size of the maximum cardinality set in $\cal F$. 
\end{enumerate}

\subsection{\srbsc\ parameterized by $k_\ell+k_r$ and $k_\ell+r$}

\begin{theorem}\label{MeCC_redn}
 \srbsc\ is W-hard  in the following cases: 
 \begin{enumerate}[i)]
  \item When every set contains at least two red elements but at most three elements, and the parameters are $\{k_\ell,k_r\}$, the problem is \W[1]-hard.
  \item When every set contains at most one red element and the parameters are $\{k_\ell,r\}$, then the problem is \W[2]-hard.
 \end{enumerate}
\end{theorem}

\begin{proof}
We start by proving the first result.
From an instance $(G=(V,E),k)$ of \MCC\ parameterized by $k$, we construct an instance $(U=(R,B),\F)$ 
of \srbsc\ parameterized by $k_\ell+k_r$ with the restriction that the size of each set is at most three and there are at least $2$ red elements. The construction is as follows. 
 \begin{itemize}
\item Let the given vertex set be $V = V_1\uplus V_2 \uplus \ldots \uplus V_k$. For every pair $(i,j)$, $1\leq i< j\leq k$, we introduce a new blue element $b_{ij}\in B$. Thus we have $k \choose 2$ blue elements.

\item For each vertex $v\in V$ we introduce a new red element $r_v \in R$.
\item $U = R \uplus B$.
\item For each $e=(u,v)\in E$ such that $u \in V_i, v \in V_j$ and $i<j$, we define a set $S_e \in \F$ which contains the elements $\{b_{ij}, r_u,r_v\}$. 
\item We set $k_r = k$ and $k_\ell = {k \choose 2}$.
\end{itemize}
This completes our construction. Notice that every set in $\F$ has at least $2$ red elements and has size exactly three. 

First, assume that $(G,k)$ is a \YES\ instance. Then there is a $k$-sized multi-colored clique $C$ in $G$. Let $E(C)$ denote the set of edges of $C$. 
Pick the subfamily $\F' = \{S_e~|~e \in E(C)\}$ of size $k \choose 2$. Since $C$ is a multi-colored clique, for all $(i,j)$, 
$ 1\leq i < j\leq k$ there is an edge $e_{ij}\in E(C)$ whose endpoints belong to $V_i$ and $V_j$. Consequently, there is a set $S_{e_{ij}} \in \F'$ that contains $b_{ij}$. The total number of red elements contained in $\F'$ is equal to the size $\vert V(C)\vert =k$. This shows that $(U,\F,k)$ is a \YES\ instance of \srbsc.

Conversely, suppose $(U,\F)$ is a \YES\ instance of \srbsc. Let $\F'$ be a minimal subfamily of at most $k \choose 2$ sets that covers at most $k$ red elements. Let $C$ be the vertices in $G$  
corresponding to the red elements in $\F'$. Notice that there are ${k}\choose{2}$ blue elements, no two of which can be covered by the same set. Thus, for all $(i,j)$, $1\leq i< j\leq k$, $\F'$ must contain exactly one set $S_e = \{b_{ij}, r^{ij}_{1},r^{ij}_{2}\}$. This implies that for every $i$, $1\leq i \leq k$ the sets in  $\F'$  must contain a red element corresponding to a vertex in 
 $V_i$. Hence, for all  $i, 1\leq i \leq k$, $C\cap V_i\neq \emptyset$. Also, $C$ forms a clique since the set $S_e = \{b_{ij}, r^{ij}_{1},r^{ij}_{2}\}$ corresponds to the edge between the vertices selected from $V_i$ and $V_j$. 
Therefore, $(G,k)$ is a \YES\ instance of \MCC.
This proves that \srbsc, parameterized by $k_\ell+k_r$, is \W[1]-hard  under the said assumption.

For the second part of the statement, observe that \SC\ is a special case of this problem and therefore, 
the problem is \W[2]-hard.
\end{proof}


\subsection{A special case of  \srbsc\ parameterized by $k_\ell$}
In this section, we restrict the input instances to those where every set has at most $1$ red element and at most $d$ blue elements. We design an \FPT\ algorithm for this special case of \srbsc\ parameterized by $k_\ell+ d$. It is reasonable to assume that there is {\em no set} in the given instance with only red elements, since Reduction Rule~\ref{redn1} can be applied to obtain an equivalent instance of \srbsc, under the parameters of $\{k_\ell,d\}$.

We were able to show that this problem has an \FPT\ algorithm. However, it was pointed out to us by an anonymous reviewer that there is a simple algorithm based on Dynamic Programming technique. Thus, we present the simpler algorithm. 

\subsubsection{A Dynamic Programming Algorithm }

 We give a Dynamic Programming algorithm to solve \srbsc\ parameterized by $k_l + d$, for the case when all sets contain at most 1 red element and at most d blue elements. Our algorithm guesses the red point that can be added to the solution one by one and also guesses the sets that can cover it and covers the remaining blue points optimally.
\begin{lem}
There exists a \FPT\ algorithm that solves \srbsc\ when each set in the input instance contains at most $1$ red element and at most $d$ blue elements. The running time of this algorithm is $O(2^{2dk_l} (\vert U \vert +\vert \F \vert )^{\Oh(1)})$.
\end{lem}
\begin{proof}
 Let $B' \subseteq B$, $r' \in R \cup nil, j \in \mathbb{N}$. Let $W[B', r']$ represent the minimum cardinality of a family $\F' \subseteq \F$ that covers all elements in $B'$ and does not cover any red element except $r'$ (no red element if $ r'$ is $nil$). The value of $W[B', r']$ is $+\infty$ if no such $\F' \subseteq \F$ exists. Let $T[B', j]$ represent the minimum cardinality of a family $\F' \subseteq \F$ that covers all elements in $B'$ and covers at most $j$ red elements. Clearly the instance is a YES instance if and only if $T[B, k_r] \leq k_l$. 
 
 We can compute the value of $T[B, k_r]$ using the following recurrence.
\\
$T[B', 0] = W[B', nil]$\\
$T[B', j] = \min _{r' \in (R~ \cup ~{nil}) } \min _{B'' \subseteq B'} (W[B'', r'] + T[B' \setminus B'', j-1])$
\\

 Similarly we can compute the value of $W[B', r']$ using the following recurrence.
 \\
 $W[\emptyset, r'] = 0 $\\
$W[B', r'] = 1 + \min_{S \in \F, S \cap R =\emptyset ~or ~ S \cap R =\{ r'\},S \cap B' \ne \emptyset} W[B' \setminus S, r']  $
\\Let us first show that the recurrence for $W$ is correct. The proof is by induction on $\vert B \vert$. When $\vert B \vert =0$ the recurrence correctly returns zero. When $\vert B \vert >0$, $W[B' \setminus S, r'] $  returns the minimum cardinality of a family $\F' \subseteq \F$ that covers all elements in $B'\setminus \F$ and does not cover any red element except $r'$ (By induction hypothesis). Therefore, $S \cup \F'$  covers all elements in $B'$ and does not cover any red element except $r'$. Since we are doing this for every $S \in \F$ and take the minimum value, the recurrence indeed returns the minimum cardinality of a family $\F' \subseteq \F$ that covers all elements in $B'$ and does not cover any red element except $r'$.

Now we show that the recurrence for $T$ is correct by induction on $j$. When $j=0$, the recurrence returns the value of $W[B,nil]$ which returns the minimum cardinality of a family $\F' \subseteq \F$ that covers all elements in $B'$ and does not cover any red element. When $j>0$, we consider a number of sets containing the same red element $r'$, paying for the blue elements $B'' \subseteq B'$ they cover, and cover the remaining blue elements $B' \setminus B''$ optimally by induction hypothesis. Since we do this for all red points and return the
 minimum value, the recurrence is correct.

\noindent \textbf{Running time: }To compute the value of $T[B, k_r]$ using the above recurrence, we have to compute at most $2^{|B|}|U|$ values of $W$ and $T$, which is at most $2^{d  k_l}  |U|$ in YES-instances. Every value of $W$ can be computed in $O(|U|)$ time using previously computed values. To compute a value of $T$, we take the minimum over all choices of $r'$ in $R$, over at most $2^{|B|} \leq 2^{d  k_l}$ choices of $B''$, and look up earlier values. Thus the running time is bounded by $O(2^{2dk_l} (\vert U \vert +\vert \F \vert )^{\Oh(1)})$.
\end{proof}

 When it comes to kernelization for this special case, we show that even for \slrbsc\ parameterized by $k_\ell+d$ there cannot be a polynomial kernel unless $\CoNP\subseteq \NP/\mbox{poly}$. For this we will give a polynomial parameter transformation from \SC\ parameterized by universe size $n$. The ppt reduction is exactly the one given in Theorem~\ref{set_cover_redn}.

 \begin{thmk}
 \slrbsc\ parameterized by $k_\ell+d$, and where every line has at most $1$ red element and at most $d$ blue elements, does not allow a polynomial kernel unless $\CoNP\subseteq \NP/\mbox{poly}$.
\end{thmk}

\section{Conclusion}

In this paper, we provided a complete parameterized and kernelization dichotomy of the \slrbsc\ problem, under all possible combinations of its natural parameters. We also studied \solrbsc\ and \srbsc\  under different parameterizations. 
The next natural step seems to be a study of the \srbsc\ problem, when the sets are hyperplanes. Another interesting variant is when the set system has bounded intersection. 

We believe that the running time of the \FPT\ algorithm for \slrbsc\ parameterized by $k_\ell,k_r$ is tight, up to the constants appearing in the exponents. It would be interesting to show that the problems cannot have algorithms with running time dependence on parameters as  $k_\ell^{o(k_\ell)}\cdot k_r^{\Oh(k_r)}$ or $k_\ell^{\Oh(k_\ell)}\cdot k_r^{o(k_r)}$, under standard complexity theoretic assumptions (like the Exponential Time Hypothesis). 
\bibliographystyle{plain}

\begin{thebibliography}{10}

\bibitem{BodlaenderDFH09}
Hans~L. Bodlaender, Rodney~G. Downey, Michael~R. Fellows, and Danny Hermelin.
\newblock On problems without polynomial kernels.
\newblock {\em J. Comput. Syst. Sci.}, 75(8):423--434, 2009.

\bibitem{CDKM00}
Robert~D Carr, Srinivas Doddi, Goran Konjevod, and Madhav~V Marathe.
\newblock On the red-blue set cover problem.
\newblock In {\em SODA}, volume~9, pages 345--353, 2000.

\bibitem{CH13}
Timothy~M Chan and Nan Hu.
\newblock Geometric red-blue set cover for unit squares and related problems.
\newblock In {\em CCCG}, 2013.

\bibitem{DemaineFHT05}
Erik~D. Demaine, Fedor~V. Fomin, Mohammad~Taghi Hajiaghayi, and Dimitrios~M.
  Thilikos.
\newblock Subexponential parameterized algorithms on bounded-genus graphs and
  \emph{H}-minor-free graphs.
\newblock {\em J. {ACM}}, 52(6):866--893, 2005.

\bibitem{DGNW07}
Michael Dom, Jiong Guo, Rolf Niedermeier, and Sebastian Wernicke.
\newblock Red-blue covering problems and the consecutive ones property.
\newblock {\em J. Discrete Algorithms}, 6(3):393--407, 2008.

\bibitem{DLS2014}
Michael Dom, Daniel Lokshtanov, and Saket Saurabh.
\newblock Kernelization lower bounds through colors and ids.
\newblock {\em ACM Trans. Algorithms}, 11(2):13:1--13:20, 2014.

\bibitem{DF99}
Rodney~G. Downey and Michael~R. Fellows.
\newblock {\em Parameterized Complexity}.
\newblock Springer-Verlag, 1999.
\newblock 530 pp.

\bibitem{FellowsKNRRSTW08}
Michael~R. Fellows, Christian Knauer, Naomi Nishimura, Prabhakar Ragde,
  Frances~A. Rosamond, Ulrike Stege, Dimitrios~M. Thilikos, and Sue Whitesides.
\newblock Faster fixed-parameter tractable algorithms for matching and packing
  problems.
\newblock {\em Algorithmica}, 52(2):167--176, 2008.

\bibitem{FlumGrohebook}
J{\"o}rg Flum and Martin Grohe.
\newblock {\em Parameterized Complexity Theory}.
\newblock Texts in Theoretical Computer Science. An EATCS Series.
  Springer-Verlag, Berlin, 2006.

\bibitem{FominGSS09}
Fedor~V. Fomin, Serge Gaspers, Saket Saurabh, and Alexey~A. Stepanov.
\newblock On two techniques of combining branching and treewidth.
\newblock {\em Algorithmica}, 54(2):181--207, 2009.

\bibitem{FortnowS11}
Lance Fortnow and Rahul Santhanam.
\newblock Infeasibility of instance compression and succinct pcps for {NP}.
\newblock {\em J. Comput. Syst. Sci.}, 77(1):91--106, 2011.

\bibitem{ImpagliazzoPZ01}
Russell Impagliazzo, Ramamohan Paturi, and Francis Zane.
\newblock Which problems have strongly exponential complexity?
\newblock {\em J. Comput. Syst. Sci.}, 63(4):512--530, 2001.

\bibitem{Karp10}
Richard~M. Karp.
\newblock Reducibility among combinatorial problems.
\newblock In {\em Proceedings of a symposium on the Complexity of Computer
  Computations}, pages 85--103, 1972.

\bibitem{KPR14}
Stefan Kratsch, Geevarghese Philip, and Saurabh Ray.
\newblock Point line cover: The easy kernel is essentially tight.
\newblock In {\em SODA}, pages 1596--1606, 2014.

\bibitem{LP05}
Stefan Langerman and Pat Morin.
\newblock Covering things with things.
\newblock {\em Discrete \& Computational Geometry}, 33(4):717--729, 2005.

\bibitem{Mar05}
D{\'a}niel Marx.
\newblock Efficient approximation schemes for geometric problems?
\newblock In {\em ESA}, pages 448--459. Springer, 2005.

\bibitem{Mar06}
D{\'a}niel Marx.
\newblock Parameterized complexity of independence and domination on geometric
  graphs.
\newblock In {\em Parameterized and Exact Computation}, pages 154--165.
  Springer, 2006.

\bibitem{MathiesonS08}
Luke Mathieson and Stefan Szeider.
\newblock The parameterized complexity of regular subgraph problems and
  generalizations.
\newblock In {\em CATS}, volume~77, pages 79--86, 2008.

\bibitem{Peleg07}
David Peleg.
\newblock Approximation algorithms for the label-cover\({}_{\mbox{max}}\) and
  red-blue set cover problems.
\newblock {\em J. Discrete Algorithms}, 5(1):55--64, 2007.

\bibitem{RS08}
Venkatesh Raman and Saket Saurabh.
\newblock Short cycles make {W}-hard problems hard: {FPT} algorithms for
  {W}-hard problems in graphs with no short cycles.
\newblock {\em Algorithmica}, 52(2):203--225, 2008.

\bibitem{Robertson1984}
Neil Robertson and P.D Seymour.
\newblock Graph minors {III} planar tree-width.
\newblock {\em Journal of Combinatorial Theory, Series B}, 36(1):49 -- 64,
  1984.

\end{thebibliography}

\end{document}